\documentclass[11pt]{article}

\usepackage{amssymb,amsfonts,amsmath,amsthm,url, placeins, amsthm}
\usepackage{multirow,ctable} 
\usepackage{paralist,xspace}
\usepackage[text={6.75in,9.5in},centering]{geometry}
\usepackage{epsfig,graphicx,wrapfig}

\usepackage{multicol}
\usepackage{color}

\definecolor{gold}{rgb}{0.85,.66,0}
\definecolor{cherry}{rgb}{0.9,.1,.2}
\definecolor{burgundy}{rgb}{0.8,.2,.2}
\definecolor{orangered}{rgb}{0.85,.3,0}
\definecolor{orange}{rgb}{0.85,.4,0}
\definecolor{olive}{rgb}{.45,.4,0}
\definecolor{lime}{rgb}{.6,.9,0}
\definecolor{green}{rgb}{.2,.7,0}
\definecolor{darkgreen}{rgb}{.1,.5,0}
\definecolor{grey}{rgb}{.4,.4,.2}
\definecolor{brown}{rgb}{.4,.2,.1}
\definecolor{blue}{rgb}{0,.0, .81}
\definecolor{bluepurple}{rgb}{.3, .0, .7}

\theoremstyle{definition}

\newtheorem{theorem}{Theorem}[section]
\newtheorem*{theorem*}{Theorem}
\newtheorem{corollary}[theorem]{Corollary}

\newtheorem{example}[theorem]{Example}
\newtheorem*{example*}{Example}
\newtheorem{proposition}[theorem]{Proposition}
\newtheorem{definition}[theorem]{Definition}
\newtheorem{lemma}[theorem]{Lemma}

\def\B{\mathcal{B}}
\def\F{\mathbb{F}}
\def\R{\mathbb{R}}
\def\C{\mathcal{C}}
\def\od{\stackrel{\mathrm{def}}{=}}
\def\U{\mathcal{U}}
\def\N{\mathcal{N}}

\def\CF{\mathrm{CF}}
\def\RR{\mathbb{R}}

\def\Lk{\operatorname{Lk}}
\def\RF{\operatorname{RF}}

\def\vv{{v}}

\begin{document}
\begin{center}
\textbf{ \Large Algebraic signatures of convex and non-convex codes}
\medskip

Carina Curto$^a$, Elizabeth Gross$^b$, Jack Jeffries$^{c,d}$, Katherine Morrison\footnote{corresponding author: Katherine Morrison, University of Northern Colorado, \url{katherine.morrison@unco.edu}}$^{,e,a}$, \\
Zvi Rosen$^{f,g,a}$, Anne Shiu$^h$, and Nora Youngs$^{i, j}$\\
\end{center}

\begin{small}
\hspace{.25in} $^a$ Department of Mathematics, The Pennsylvania State University, University Park, PA 16802

\hspace{.25in} $^b$ Department of Mathematics, San Jos\'{e} State University, San Jos\'{e}, CA 95192

\hspace{.25in} $^c$ Department of Mathematics, University of Utah, Salt Lake City, UT 84112

\hspace{.25in} $^d$ Department of Mathematics, University of Michigan, Ann Arbor, MI 48109

\hspace{.25in} $^e$ School of Mathematical Sciences, University of Northern Colorado, Greeley, CO 80639

\hspace{.25in} $^f$ Department of Mathematics, University of California, Berkeley, Berkeley, CA 94720

\hspace{.25in} $^g$ Departments of Mathematics \& Biology, University of Pennsylvania, Philadelphia, PA 19104

\hspace{.25in} $^h$ Department of Mathematics, Texas A\&M University, College Station, TX 77843

\hspace{.25in} $^i$ Department of Mathematics, Harvey Mudd College, Claremont, CA 91711

\hspace{.25in} $^{j}$ Department of Mathematics and Statistics, Colby College, Waterville, Maine 04901

\end{small}

\medskip

\section*{Abstract}
A \emph{convex code} is a binary code generated by the pattern of intersections of a collection of open convex sets in some Euclidean space.  Convex codes are relevant to neuroscience as they arise from the activity of neurons that have convex receptive fields.  In this paper, we use algebraic methods to determine if a code is convex.  Specifically, we use the \emph{neural ideal} of a code, which is a generalization of the \emph{Stanley-Reisner ideal}.  Using the neural ideal together with its standard generating set, the \emph{canonical form}, we provide algebraic signatures of certain families of codes that are non-convex.  We connect these signatures to the precise conditions on the arrangement of sets that prevent the codes from being convex.  Finally, we also provide algebraic signatures for some families of codes that are convex, including the class of \emph{intersection-complete codes}.  These results allow us to detect convexity and non-convexity in a variety of situations, and point to some interesting open questions.

\medskip
\noindent \textbf{Keywords:} neural coding, convex codes, neural ideal, local obstructions, simplicial complexes, links

\begin{small}
\tableofcontents
\end{small}

\section{Introduction} \label{sec:intro}

A \emph{convex code} is a binary code generated by the pattern of intersections of a collection of open convex sets in some Euclidean space (see Section~\ref{sec:background} for a precise definition and example).  Convex codes have been experimentally observed in sensory cortices \cite{HubelWiesel59} and hippocampus \cite{Okeefe}, where they arise from convex \emph{receptive fields}; this connection has previously been described in detail in \cite{CurrEventsBulletin, MRC, neural_ring}.  Given their relevance to neuroscience, it is valuable to further understand the intrinsic structure of convex codes.  In particular, how can we detect if a neural code is convex?

We have previously found combinatorial constraints that must be satisfied by any code that is convex \cite{MRC}.  In this work, we further address the question of convexity via an algebraic object known as the \emph{neural ideal} $J_\mathcal{C}$, first introduced in \cite{neural_ring}, which is a generalization of the well-studied \emph{Stanley-Reisner ideal}.  We first present conditions, which we refer to as \emph{algebraic signatures}, on $J_\C$ and its standard generating set the canonical form $\CF(J_\mathcal{C})$ that detect that a code is not convex.  We also connect these signatures to precise conditions on the arrangement of sets that prevent a code from being convex.  Finally, we also provide algebraic signatures of certain combinatorial families of convex codes, including \emph{intersection-complete codes}, first introduced in \cite{intersection-complete}.  

In Section~\ref{sec:background}, we provide some background on the algebra of neural codes, convexity of codes, receptive field relationships, and \emph{local obstructions} to convexity.
Next, Section~\ref{sec:main-results} highlights the main results of the paper.  Specifically, Theorem~\ref{thm:A-sigs-non-convex} provides algebraic signatures of two classes of local obstructions; codes satisfying these signatures are thus guaranteed to be non-convex.  Theorem~\ref{thm:int-complete} gives an algebraic signature for the class of intersection-complete codes, which have been proven to be convex.  Section~\ref{sec:examples-main-results} illustrates these main results through a series of example codes satisfying these algebraic signatures.  

The remainder of the paper is organized as follows: Section~\ref{sec:local-obs} formalizes the notion of local obstruction, and Section~\ref{sec:non-convex} provides further results on detecting local obstructions algebraically, including the proof of Theorem~\ref{thm:A-sigs-non-convex}.  Section~\ref{sec:detecting-convex} focuses on algebraic signatures guaranteeing convexity, and includes the proof of Theorem~\ref{thm:int-complete}.  Finally, Section~\ref{sec:examples} collects all the algebraic signatures presented in this paper and provides additional examples of codes satisfying these signatures.  

\subsection{Background}\label{sec:background}
In this paper, we develop algebraic tools for analyzing neural codes, which are collections of binary patterns.  A {\it binary pattern} on $n$ neurons is a string of $0$s and $1$s of length $n$, with a $1$ for each active neuron and a $0$ denoting silence.  We can also view a binary pattern as the subset of active neurons $\sigma \subseteq [n]\od \{1,\ldots,n\}$, so that $i \in \sigma$ precisely when there is a $1$ in the $i$th entry of the binary pattern; thus, we will consider 0/1 strings of length $n$ and subsets of $[n]$ interchangeably.  For example, $1011$ and $0100$ are also denoted $\{1,3,4\}$ and $\{2\}$, respectively.

A {\it neural code} on $n$ neurons, $\C \subseteq 2^{[n]}$, is a collection of binary patterns.  Such a code is also referred to as a {\it combinatorial code} in the neuroscience literature  \cite{neuro-coding}.
The elements of a code are called {\it codewords}.
For convenience, we will always assume a neural code $\C$ includes the all-zeros codeword, $00\cdots0 \in \C$; the presence or absence of the all-zeros codeword has no effect on the code's convexity (see Definition~\ref{def:convex}, below), which is the main focus of this paper.

\subsubsection*{Algebra of neural codes} 
In order to represent a neural code algebraically, it is useful to consider binary patterns of length $n$ as elements of $\F_2^n$, where $\F_2$ is the finite field of two elements: $0$ and $1$.  Polynomials $f \in \F_2[x_1,\ldots,x_n]$ can be evaluated on a binary pattern of length $n$ by evaluating each indeterminate $x_i$ at the 0/1 value of the $i^{\mathrm{th}}$ neuron. For example, if $f = x_1 x_3 (1-x_2) \in \F_2[x_1,\ldots,x_4]$, then $f(1011) = 1$ and $f(1100) = 0$.  

It is natural to then consider the ideal 
$$I_\C \od \{ f \in \F_2[x_1,\ldots,x_n]~|~f(c)=0,~\forall~c \in \C\}$$
of polynomials that vanish on a neural code $\C$.  However, this ideal contains extraneous \emph{Boolean relations} $\B=\langle x_i(1-x_i) \rangle$ that do not capture any information specific to the code.  Thus we turn instead to the neural ideal $J_\C$, first introduced in \cite{neural_ring}, which captures all the information in $I_\C$ that is specific to the code, thus omitting the Boolean relations.  More precisely, the neural ideal can be defined in terms of characteristic functions of non-codewords:
$$J_\C \od \langle \chi_{\vv}~|~\vv \in \F_2^n \setminus \C \rangle$$
where $\chi_{\vv}$ is the characteristic function
\begin{eqnarray}\label{eq:char-pseudo-monomials}
\chi_{\vv} \od \hspace{-.1in}\displaystyle\prod_{\{i | v_i=1\}} \hspace{-.1in} x_i \prod_{\{j | v_j=0\}} \hspace{-.1in}(1-x_j).
\end{eqnarray}
Note that the variety of both $I_\C$ and $J_\C$ is precisely the code $\C$ \cite{neural_ring}.  

The characteristic functions used to define the neural ideal are examples of {\em pseudo-monomials}, polynomials $f\in\F_2[x_1, \ldots,x_n]$ that can be written in the form 
$$f =  x_\sigma \prod_{j\in \tau} (1-x_j),$$
where $x_\sigma \od  \prod_{i \in \sigma} x_i$ and $\sigma,\tau \subset [n]$ with $\sigma\cap \tau=\emptyset$. 
Pseudo-monomials in $J_\C$ come in two types\footnote{There is a third type (see \cite{neural_ring}), but this is eliminated by our convention that $00\cdots0\in \C$.}:
\begin{itemize}
\item Type 1: $x_\sigma$, for $\sigma \neq \emptyset$, and
\item Type 2: $x_\sigma \prod_{i \in \tau} (1-x_i)$, for $\sigma,\tau \neq \emptyset, \textrm{ with } \sigma \cap \tau =\emptyset$.
\end{itemize}

For any ideal $J \subseteq \F_2[x_1,\ldots,x_n]$, a pseudo-monomial $f \in J$ is called \emph{minimal} if there does not exist another pseudo-monomial $g \in J$ with $\deg(g) < \deg(f)$ such that $f = hg$ for some $h \in \F_2[x_1,\ldots,x_n]$.  If $J$ is an ideal generated by a set of pseudo-monomials, 
the \emph{canonical form} of $J$ is the set of all minimal pseudo-monomials of $J$:
$$\CF(J) \od \{ f \in J \mid f \text{ is a minimal pseudo-monomial}\}.$$ 
For any neural code $\C$, the neural ideal $J_\C$ is generated by pseudo-monomials, and thus has a canonical form $\CF(J_\C)$.\footnote{Furthermore, every ideal generated by pseudo-monomials is actually the neural ideal of some neural code \cite{mo-student-thesis}.}
We denote the Type 1 and Type 2 pseudo-monomials of $\CF(J_\C)$ by $\CF^1(J_\C)$ and $\CF^2(J_\C)$, respectively, so that:
$$\CF(J_\C)= \CF^1(J_\C) \cup \CF^2(J_\C).$$

\begin{example}\label{ex:example1}
Consider the code $\C=\{0000,~0100,~0010,~0001,~1100,~1010,~0110,~1011\}$.  The neural ideal $J_\C$ is given by 
\begin{eqnarray*}
J_\C &=& \langle x_1(1-x_2)(1-x_3)(1-x_4), x_1x_4(1-x_2)(1-x_3), x_2x_4(1-x_1)(1-x_3), \\
&& x_3x_4(1-x_1)(1-x_2), x_2x_3x_4(1-x_1), x_1x_2x_4(1-x_3), x_1x_2x_3(1-x_4), x_1x_2x_3x_4 \rangle,
\end{eqnarray*}
which has canonical form $\CF(J_{\C})=\CF^1(J_{\C}) \cup \CF^2(J_{\C})$, where
$$ \CF^1(J_{\C})=\{x_1x_2x_3, x_2x_4\}  \hspace{.03in}\textrm{ and } \hspace{.03in} \CF^2(J_{\C})=\{x_1(1-x_2)(1-x_3),~x_1x_4(1-x_3),~x_3x_4(1-x_1)\}.$$
\end{example}

Note that $\CF(J_\C)$ is a generating set for $J_\C$, as every pseudo-monomial of $J_\C$ is a multiple of an element in $\CF(J_\C)$.  Furthermore, $\CF^1(J_\C)$ generates the ideal of monomials in $J_\C$, which is precisely the \emph{Stanley-Reisner ideal} of the associated simplicial complex $\Delta(\C)$, where
$$\Delta(\C) \od \{\sigma \subseteq [n] \mid \sigma \subseteq c \text{ for some } c \in \C\}$$
is the smallest abstract simplicial complex on $[n]$ that contains all elements of $\C$ \cite{neural_ring}. In particular, if $\C$ is a simplicial complex, then $J_\C$ is precisely the Stanley-Reisner ideal of $\C$.  Note that the \emph{facets} of $\Delta(\C)$, which are maximal elements of the simplicial complex under inclusion, correspond to the maximal codewords of $\C$.

The canonical form of a code $\C$ can be computed algorithmically; for example, \cite[Section 4.5]{neural_ring} provides an algorithm using primary decompositions of pseudo-monomial ideals.  A more efficient algorithm has since been proposed in \cite{NeuralIdealsSage}, with software publicly available \cite{github}.  Supplemental Text S1 gives full details for computing the canonical form of an example code by hand; for information on using software to compute $\CF(J_{\C})$, see \cite{NeuralIdealsSage}.

\subsubsection*{The code of a cover}
 Let $X$ be a topological space.  A collection of non-empty open sets $\U = \{U_1,\ldots,U_n\}$, where each $U_i \subset X$,
is called an {\it open cover}.  Given an open cover $\U$, the {\it code of the cover} is the neural code
$$\C(\U) \od \{ \sigma \subseteq [n]~\vert~U_\sigma \setminus \bigcup_{j \in [n] \setminus \sigma} U_j \neq \emptyset \},
$$ 
where $U_\sigma \od \bigcap_{i \in \sigma} U_i$.  We say that a code $\C$ is \emph{realized by} $\U$ if $\C=\C(\U)$.  Observe that $X$ is subdivided into regions defined by intersections of the open sets in $\U$.  
Each codeword in $\C(\U)$ then corresponds to a non-empty intersection that is not covered by other sets in $\U$ (see Example~\ref{ex:convex}).
By convention, the empty intersection $U_\emptyset = \bigcap_{i \in \emptyset} U_i$ equals $X$, so that $\emptyset \in \C(\U)$ if and only if $\bigcup_{i \in [n]} U_i \subsetneq X$.  We will assume $\bigcup_{i \in [n]} U_i \subsetneq X$, so that $00\cdots 0 \in \C$ (i.e., $\emptyset \in \C$), in agreement with our convention.  

It is important to note that $\C(\U)$ is not the same as the {\em nerve} $\N(U)$ of the cover, which consists of all non-empty intersections, regardless of whether the intersection region is covered by other sets:
$$\N(\U) \od \{\sigma \subseteq [n] \mid U_\sigma \neq \emptyset \}.$$
In fact, $\N(\U)= \Delta(\C(\U))$, the simplicial complex of the code \cite{neural_ring}.  The nerve of any cover $\U$ such that $\C = \C(\U)$ can thus be recovered directly from the code as $\Delta(\C)$, without reference to a specific cover.  The code $\C(\U)$, however, contains additional information about $\U$ that is not captured by the nerve alone (see \cite[Section 2.3.2]{neural_ring}).

\begin{example}\label{ex:convex}
Consider the configuration of sets $\U = \{U_1, \ldots, U_4\}$ shown in Figure~\ref{fig:convex-code}.  The code of the cover is $\C=\C(\U) = \{ 0000,~1000,~0100,~0010,~1100,~1001,~0110,~0101,~1101\}.$ Note that from $\C$ alone, we can detect that any realization must have $U_4 \subseteq U_1 \cup U_2$, since every codeword with a 1 in the 4th position has a 1 in the 1st or 2nd position as well.  However, this containment information is not available from the nerve $\N(\U)$.  
\begin{figure}[!ht]
\begin{centering}
\includegraphics[height=1.5in]{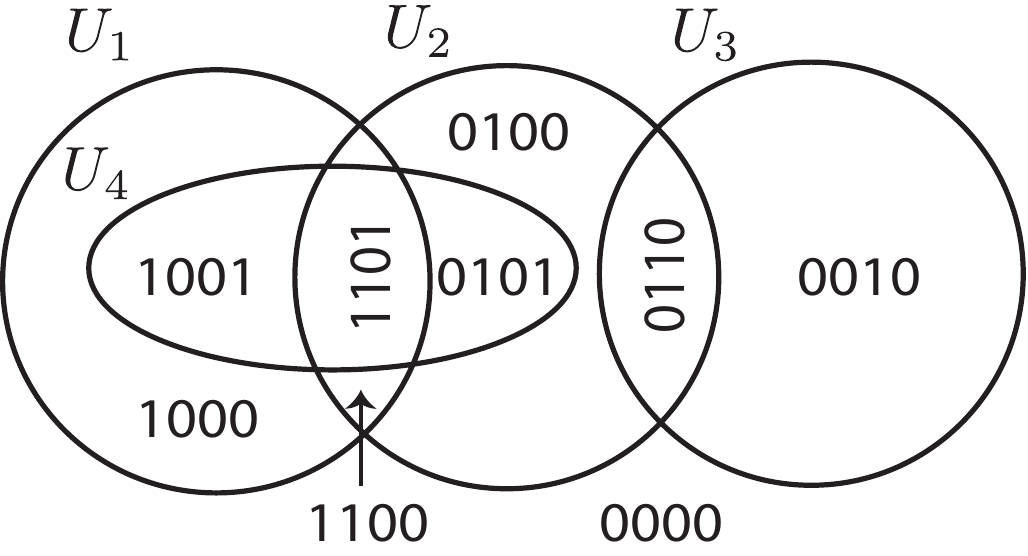}
\vspace{-.1in}
\caption{Code of the cover $\U = \{U_1, \ldots, U_4\}$.}
\label{fig:convex-code}
\end{centering}
\vspace{-.15in}
\end{figure} 
\end{example}

\paragraph{RF relationships and the neural ideal.}
Any realization of a code $\C$ by an open cover will satisfy relationships among the $U_i$ that are intrinsic to the code itself.  Because of the neuroscience motivation, where the $U_i$ model receptive fields, we call these \emph{receptive field relationships} \cite{neural_ring}.

\begin{definition}
For $\sigma, \tau \subseteq [n]$ with $\sigma \neq \emptyset$ and $\sigma \cap \tau = \emptyset$, we say that $(\sigma, \tau)$ is a \emph{receptive field (RF) relationship} of a code $\C$ if 
$$U_\sigma \subseteq \bigcup_{i \in \tau} U_i \, \, \, \,  \text{ and } \, \, \, U_\sigma \cap U_i \neq \emptyset \text{ for all } i \in \tau,$$ 
for any $\U = \{U_1, \ldots, U_n\}$ where $\C=\C(\U)$.  $\RF(\C)$ denotes the collection of RF relationships of $\C$.
\end{definition}

It is important to note that the receptive field relationships $\RF(\C)$ are strictly a function of the code itself and do not depend on any particular realization of $\C$ as $\C(\U)$.  Specifically, RF relationships correspond to pseudo-monomials in $J_\C$ as shown in Table~\ref{table:J_C-and-RF}, and thus are detectable algebraically without reference to a specific cover $\U$ \cite{neural_ring}.

\begin{table}[!ht] \label{table:types}
\begin{center}
\vspace{.1in}
\begin{small}
\begin{tabular}{llcl}
Relation type& \hspace{.08in}Pseudo-monomial && \hspace{-.1in}RF condition\\
 \specialrule{.2em}{.2em}{.2em} 
\hspace{.17in}Type 1 & $x_\sigma \in J_\C$ & \phantom{hi}$\Leftrightarrow$ \phantom{hi}& $U_\sigma = \emptyset$\\
 \specialrule{.05em}{.2em}{.2em} 
\hspace{.17in}Type 2 & $x_\sigma\prod_{i \in \tau} (1-x_i) \in J_\C$ & \phantom{hi}$\Leftrightarrow$ \phantom{hi}& $U_\sigma \subseteq \bigcup_{i \in \tau} U_i$ \\
\specialrule{.125em}{.2em}{.2em} 
 \end{tabular}
\caption{Types of pseudo-monomials in $J_\C$ and the corresponding conditions on receptive fields.  Note that the presence of a Type 2 pseudo-monomial $x_\sigma\prod_{i \in \tau} (1-x_i) \in J_\C$ is not sufficient to guarantee that $(\sigma, \tau)$ is actually an RF relationship.  Such a pseudo-monomial ensures the covering relationship $U_\sigma \subseteq \bigcup_{i \in \tau} U_i$, but to guarantee that $(\sigma, \tau) \in \RF(\C)$ for $\tau \neq \emptyset$ we must also have $x_\sigma x_i \notin J_\C$ for all $i \in \tau$.
}
\label{table:J_C-and-RF}
\vspace{-.15in}
\end{small}
\end{center}
\end{table}

The RF relationships of the form $(\sigma, \emptyset)$ capture when $U_\sigma =\emptyset$, and thus $\sigma \notin \N(\U)$, yielding a complete description of $\N(\U) = \Delta(\C)$.  In contrast, the RF relationships $(\sigma, \tau)$ for $\tau \neq \emptyset$ capture when an intersection is covered so that $\sigma \notin \C$ despite $\sigma \in \Delta(\C)$, thus measuring how $\C$ deviates from its simplicial complex.  

A RF relationship $(\sigma, \tau)$ is called \emph{minimal} if no neuron can be removed from $\sigma$ or $\tau$ without destroying the containment $U_\sigma \subseteq \bigcup_{i \in \tau} U_i$.   The following useful fact is a direct consequence of \cite[Theorem 4.3]{neural_ring}, which allows us to interpret the elements of $\CF(J_\C)$ as minimal RF relationships.

\begin{lemma}\label{lemma:minRF}
The pseudo-monomial $x_\sigma \prod_{i \in \tau} (1-x_i) \in \CF(J_\C)$ if and only if  $(\sigma,\tau)$ is a minimal RF relationship of $\C$.
\end{lemma}

Thus, the canonical form gives a compact description of $J_\mathcal{C}$ that captures all the minimal intersection and containment relations that must exist among sets that give rise to the code.

\subsubsection*{Convex codes}
When the open cover $\U$ is contained in $\R^d$ for some $d$, the sets $U_i$ may (for some codes) be chosen to all be convex.  If this is possible, we say that the code is \emph{convex}:

\begin{definition} \label{def:convex}
Let $\C$ be a neural code on $n$ neurons.  If there exists an open cover $\U = \{U_1,\ldots,U_n\}$ such that $\C = \C(\U)$ and every $U_i$ is a convex subset of $\RR^d$ for a fixed $d$, then we say that $\C$ is {\em convex}. 
\end{definition}

Note that the code in Example~\ref{ex:convex} is convex since it can be realized via the convex sets shown in Figure~\ref{fig:convex-code}.  In contrast, the code from Example~\ref{ex:example1} is \underline{not} convex, as the following example shows.

\begin{example} \label{ex:code-local-obs}
Recall the code $\C=\{0000,~0100,~0010,~0001,~1100,~1010,~0110,~1011\}$ from Example~\ref{ex:example1}.  Neuron 1 always co-fires with neuron 2 or neuron 3 since a 1 only occurs in the first entry when it is accompanied by a 1 in the second or third entry.  This forces the RF relationship $U_1 \subseteq U_2 \cup U_3$ to hold in any realization of the code.  But neurons 1, 2, and 3 never co-fire, so $U_1\cap U_2\cap U_3 = \emptyset$.  Thus $U_1$ is the disjoint union of non-empty open sets $U_1 \cap U_2$ and $U_1 \cap U_3$, and so $U_1$ is disconnected.  Since any convex set is connected, we conclude that $U_1$ cannot be convex, and thus $\C$ is \underline{not} convex.  
\end{example}

This topological mismatch between the underlying set $U_1$ and its cover by $U_1 \cap U_2$ and $U_1 \cap U_3$ is an example of a \emph{local obstruction} \cite{MRC, no-go}; we define local obstructions precisely in Section~\ref{sec:local-obs}.  Notice that this local obstruction is immediately identifiable from the canonical form $\CF(J_\C)$ seen in Example~\ref{ex:example1}: the RF relationship $U_1 \subseteq U_2 \cup U_3$ is detectable from $x_1(1-x_2)(1-x_3) \in \CF^2(J_\C)$ and the RF relationship $U_1\cap U_2\cap U_3 = \emptyset$ is captured by $x_1x_2x_3 \in \CF^1(J_\C)$.  

\subsection{Summary of main results}\label{sec:main-results}
\paragraph{Detecting non-convex codes.}\label{sec:main-results-nonconvex}
Example~\ref{ex:code-local-obs} shows that some local obstructions to convexity can be detected algebraically from the neural ideal of a code.  In particular, any code satisfying the \emph{algebraic signature} $x_\sigma(1-x_i)(1-x_j)\in \CF^2(J_\C)$ and $x_\sigma x_i x_j \in \CF^1(J_\C)$ is guaranteed to be non-convex.  This is because $U_\sigma$ is forced to be disconnected since it is the disjoint union of the nonempty sets $U_\sigma \cap U_i$ and $U_\sigma \cap U_j$.  

Theorem~\ref{thm:A-sigs-non-convex} gives two additional algebraic signatures of local obstructions that force a code to be non-convex.  The first signature captures more generally when the nerve of a cover of $U_\sigma$ is disconnected, thus forcing $U_\sigma$ to be disconnected and non-convex.  The second signature captures cases when the nerve is a hollow simplex, thus forcing $U_\sigma$ to contain a hole.  In other words, these signatures capture when the nerve of the cover of $U_\sigma$ has a nontrivial 0th homology group and nontrivial top homology group, respectively.  It remains an open question to identify algebraic signatures that can detect when a relevant nerve has an \emph{intermediate} homology group that is nontrivial.

\begin{theorem}\label{thm:A-sigs-non-convex}
Let $\C$ be a code with neural ideal $J_\C$ and canonical form
 $\CF(J_\C) = \CF^1(J_\C)~\cup~\CF^2(J_\C)$, and let $G_{\C}(\sigma,\tau)$ be the simple graph on vertex set $\tau$ with edge set $\{ (ij) \in \tau \times \tau \mid x_\sigma x_i x_j \notin J_{\C} \}$.  The following algebraic signatures imply that $\C$ is \underline{not} convex.
\begin{table}[!h]
\begin{center}
\begin{small}
\begin{tabular}{l l c l}
& Algebraic signature of $J_\C$ & & Property of $\C$\\
 \specialrule{.125em}{.6em}{.6em} 
(i) & $\exists \; x_\sigma\prod_{i \in \tau}(1-x_i) \in \CF^2(J_\C)$ s.t. $G_\C(\sigma,\tau)$ is disconnected & $\Rightarrow$ & non-convex\\
 \specialrule{.05em}{.6em}{.6em} 
(ii) &  $\exists \; x_\sigma\prod_{i \in \tau}(1-x_i) \in \CF^2(J_\C)$ s.t. $x_\sigma x_\tau \in \CF^1(J_\C)$ & $\Rightarrow$& non-convex\\
 \specialrule{.125em}{.6em}{.6em} 
\end{tabular}
\caption{Algebraic signatures of non-convex codes.  
}
\label{table:A-sigs-non-convex}
\end{small}
\end{center}
\vspace{-.25in}
\end{table}
\end{theorem}

It is important to note that although signature (i) in Table~\ref{table:A-sigs-non-convex} requires the construction of a graph based on the absence of pseudo-monomials from all of $J_\C$, this condition can actually be checked in a straightforward manner from $\CF^1(J_\C)$ alone (see Lemma~\ref{lemma:graph} in Section~\ref{thm:A-sigs-non-convex}).  The signatures of local obstructions in Theorem~\ref{thm:A-sigs-non-convex} can thus be directly detected from the canonical form of the code. The proof of Theorem~\ref{thm:A-sigs-non-convex} is given in Section~\ref{sec:non-convex}.   

Our previous work has given an alternative method of identifying the \underline{full} set of local obstructions; however, the recasting of those local obstructions in terms of RF relationships is less well understood.  A characterization of the full set of local obstructions of a code is given in Theorem 1.3 of \cite{MRC}.  In general, however, the absence of local obstructions does \underline{not} guarantee that $\C$ is convex \cite{counterexample}.  Thus, it is essential to have other methods of identifying convexity.

\paragraph{Detecting convex codes.}\label{sec:main-results-convex}
Currently the only known method for proving a code is convex is to produce a convex realization or establish that it belongs to a combinatorial family of codes for which a construction of a convex realization is known.  In the following, we give algebraic signatures for identifying when a code belongs to any of four combinatorial families of codes for which convex constructions are known.  

The simplest algebraic signatures of families of convex codes are $\CF^1(J_\C)=\emptyset$ or $\CF^2(J_\C)=\emptyset$.  Since $\CF^1(J_\C)$ captures minimal subsets missing from $\Delta(\C)$, the signature $\CF^1(J_\C)=\emptyset$ implies $\Delta(\C)$ is the full simplex, and so $\C$ must contain the all-ones word.  Convex realizations of such codes were given in \cite{MRC}.  When $\C$ contains the all-ones word ($\CF^1(J_\C)=\emptyset$), $\Delta(\C)$ has a single facet, and this fact is exploited in the construction of convex realizations of these codes.  More generally, if $\Delta(\C)$ has disjoint facets this same construction can be employed in parallel for each facet, ensuring these codes are also convex \cite{MRC}.  These codes can also be detected algebraically, but the signature is more complicated, so we save the statement and proof of the signature for Section \ref{sec:detecting-convex}.

On the other hand, $\CF^2(J_\C)=\emptyset$ implies that $\C$ is a simplicial complex, which is guaranteed to have a convex realization \cite{MRC, Tancer-survey}.  These codes can be generalized to a broader family of codes known as \emph{intersection-complete} codes, which are also known to be convex \cite{intersection-complete}.

\begin{definition}
A code $\C$ is \emph{intersection-complete} ($\cap$-complete) if every intersection of codewords is also a codeword in $\C$; i.e.\ $\sigma, \omega \in \C$ implies that $\sigma \cap \omega \in \C$.  
\end{definition}

The algebraic signature for $\cap$-complete codes is given in the following theorem, whose proof appears in Section~\ref{sec:detecting-convex}.

\begin{theorem}\label{thm:int-complete}
A code $\C$ is $\cap$-complete if and only if every pseudo-monomial $x_\sigma \prod_{i \in \tau} (1-x_i) \in \CF(J_\C)$ has $|\tau| \leq 1$.  If $\C$ is $\cap$-complete, then $\C$ is convex.  
\end{theorem}

Note that if $|\tau|=0$ for all elements of $\CF(J_\C)$, then $\CF^2(J_\C)=\emptyset$, which is the signature for simplicial complex codes.  Using Table~\ref{table:J_C-and-RF}, the algebraic signature in Theorem~\ref{thm:int-complete} can be reinterpreted in terms of receptive fields as follows: for any realization of an $\cap$-complete code $\C=\C(\U)$, every intersection $U_\sigma$ for $\sigma \notin \C$ is minimally covered by a single set $U_i$ for some $i \notin \sigma$.

The families of codes presented above, for which we have algebraic signatures, are special cases of \emph{max $\cap$-complete} codes: codes for which every intersection of a collection of facets of $\Delta(\C)$ is also a codeword in $\C$.   In \cite{intersection-complete}, convex realizations of max $\cap$-codes were constructed, guaranteeing their convexity.  

\begin{theorem} \cite[Theorem 4.4]{intersection-complete}\label{thm:max-int-complete}
If a code $\C$ is max $\cap$-complete, then $\C$ is convex.
\end{theorem}

Finding an algebraic signature of max $\cap$-complete codes remains an open question.  Given that these codes generalize $\cap$-complete codes, one might hope to generalize the algebraic signature of $\cap$-complete codes to obtain a signature for this broader class.  One natural generalization is the class of codes for which every pseudo-monomial $x_\sigma \prod_{i \in \tau} (1-x_i) \in \CF(J_\C)$ has $|\tau| \leq 2$.  Unfortunately, Example~\ref{ex:diff-generalizations} (below) shows that codes with this property need not be max $\cap$-complete and vice versa.  In particular, the code in Example~\ref{ex:diff-generalizations}(b) has the $|\tau|\leq 2$ property, but is not even convex.  

\begin{example}\label{ex:diff-generalizations}
(a)  Consider the code $$\C_1=\{0000, 0100, 0010, 0001, 1100, 1010, 1001, 0110, 0011, 1110, 1011\}$$ with maximal codewords 1110 and 1011.  This code is max $\cap$-complete because it would in fact be a simplicial complex except that it is missing 1000, which is not an intersection of maximal codewords.  However, $\C_1$ does \underline{not} satisfy $|\tau|\leq 2$, since $\CF^2(J_{\C_1})= \{ x_1(1-x_2)(1-x_3)(1-x_4)\}$.\\

\noindent (b) Consider the code $$\C_2 = \{00000, 00100, 00010, 10100, 10010, 01100, 00110, 00011, 11100, 10110, 10011, 01111\}$$ with maximal codewords $11100, 10110, 10011, 01111$ (i.e.\ facets $\{123, 134, 145, 2345\}$). 
$\C_2$ is \underline{not} max $\cap$-complete since it does not contain the triple intersection of facets $1= 123\cap 134\cap 145$.  However, $\C_2$ satisfies $|\tau|\leq 2$ for all $x_\sigma\prod_{i\in \tau}(1-x_i) \in \CF^2(J_{\C_2})$ since 
$$\CF^2(\C_2) = \{x_2(1-x_3), \ x_5(1-x_4), \ x_2x_4(1-x_5), \ x_3x_5(1-x_2), \ x_1(1-x_3)(1-x_4)\}.$$
Interestingly, this code is \underline{not} convex, although it has no local obstructions \cite{counterexample}.

Note that the code from Example~\ref{ex:code-local-obs} also satisfies $|\tau|\leq 2$ and is not convex, but it has a local obstruction.  Thus, the signature $|\tau|\leq 2$ does not ensure convexity or provide guarantees about the presence/absence of local obstructions.
\end{example}

\subsection{Examples illustrating main results}\label{sec:examples-main-results}
This section gives examples of codes satisfying each of the algebraic signatures presented in Theorems~\ref{thm:A-sigs-non-convex} and \ref{thm:int-complete} together with an analysis of the implications of these signatures for RF relationships. 

We begin with an example of a code on $5$ neurons that satisfies the first signature in Theorem~\ref{thm:A-sigs-non-convex}.

\begin{example}[Theorem~\ref{thm:A-sigs-non-convex}, signature (i)]\label{ex:graph-sig}
Consider the code 
\begin{eqnarray*}
&\C=&\{00000, 11100, 10011, 01111\} \cup~ \{\textrm{all binary patterns with exactly two } 1s\}.
\end{eqnarray*}

\noindent This code has
$\CF^1(J_{\C})=\{x_1x_2x_4, x_1x_2x_5, x_1x_3x_4, x_1x_3x_5\}$ and
\begin{eqnarray*}
\hspace{-.28in} \CF^2(J_{\C})&=&\{ x_{i_1}(1-x_{i_2})(1-x_{i_3})(1-x_{i_4})(1-x_{i_5})~|~i_1,\ldots, i_5 \in [5]\}\\
&& \cup~\{ x_{i_1}x_{i_2}x_{i_3}(1-x_{i_4})~|~i_1,\ldots, i_4 \in [5]\setminus \{1\}\},
\end{eqnarray*}
\noindent where all the indices in the pseudo-monomials of $\CF^2(J_{\C})$ are distinct.
Consider 
$$x_1(1-x_2)(1-x_3)(1-x_4)(1-x_5) \in \CF^2(J_{\C}),$$ where $\sigma=\{1\}$ and $\tau=\{2,3,4,5\}$.  We will construct the graph $G=G_{\C}(\sigma, \tau)$ whose vertices are precisely the elements of $\tau$.  By definition, whenever $x_\sigma x_ix_j \notin J_{\C}$ for $i, j \in \tau$, then $(ij)$ is an edge in $G$.  Using $\CF^1(J_{\C})$, we immediately see that $(24), (25), (34),$ and $(35)$ are \emph{not} edges in $G$, and that $(23)$ and $(45)$ \emph{are} edges in $G$ (see Lemma~\ref{lemma:graph}).  Thus $G$ consists only of two disjoint edges, and is disconnected.  (Note that this implies that $U_1 \cap (U_2 \cup U_3)$ and $U_1 \cap (U_4 \cup U_5)$ are disjoint, and so $U_1$ is disconnected, as it is covered by the disjoint union of nonempty open sets.)  Therefore, signature (i) of Theorem~\ref{thm:A-sigs-non-convex} is satisfied and $\C$ is \underline{not} convex.
\end{example}

The next example gives a code on $4$ neurons satisfying the second signature of Theorem~\ref{thm:A-sigs-non-convex}.

\begin{example}[Theorem~\ref{thm:A-sigs-non-convex}, signature (ii)]\label{ex:hollow-simplex-sig}
Consider $\C=\{0000, 1110, 1101, 1011, 0111, 1100, 1010, 1001 \}$.  Then
$$\CF^1(J_{\C})=\{x_1x_2x_3x_4\} \hspace{.01in}\textrm{ and } \hspace{.01in} \CF^2(J_{\C})=\{ x_i(1-x_1)(1-x_j)~|~ i, j = 2,3,4;~i \neq j\}~\cup~\{x_1(1-x_2)(1-x_3)(1-x_4)\}.$$
Since $x_1(1-x_2)(1-x_3)(1-x_4) \in \CF^2(J_{\C})$ and $x_1x_2x_3x_4 \in \CF^1(J_{\C})$, we see that signature (ii) of Theorem~\ref{thm:A-sigs-non-convex} applies.  Thus $\C$ is \underline{not} convex.

To see the obstruction to convexity here, note that since $x_1(1-x_2)(1-x_3)(1-x_4) \in \CF^2(J_{\C})$ we have from Table~\ref{table:J_C-and-RF} that $U_1$ is minimally covered by $U_2 \cup U_3\cup U_4$.  Also, since $x_1x_2x_3x_4 \in \CF^1(J_{\C})$, the full intersection $U_1 \cap U_2 \cap U_3 \cap U_4$ is empty, but the minimality of elements in $\CF(J_\C)$ guarantees that every other intersection is non-empty.  This forces $U_1$ to contain a hole (see Figure~\ref{fig:hole}), and so $U_1$ cannot be convex, and hence $\C$ cannot be convex.

\begin{figure}[!ht]
\begin{centering}
\includegraphics[height=1.5in]{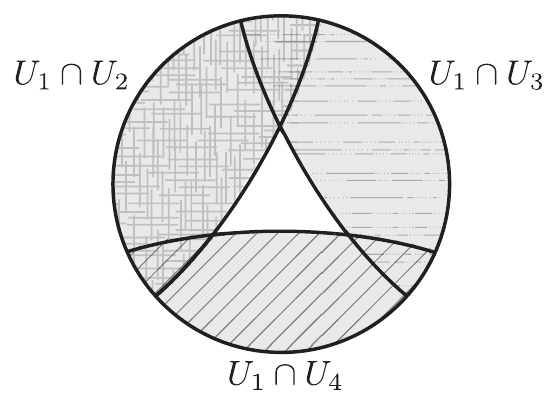}
\vspace{-.15in}
\caption{For the code in Example~\ref{ex:hollow-simplex-sig}, the set $U_1$ is the union of the shaded regions shown since it is covered by $(U_1 \cap U_2) \cup (U_1 \cap U_3) \cup (U_1 \cap U_4)$.  $U_1$ must contain a hole since the covering sets all pairwise intersect, but the full intersection is missing. }
\label{fig:hole}
\end{centering}
\end{figure} 
\end{example}

Finally, the following example shows how to use the neural ideal to detect that a code is $\cap$-complete, and thus convex.

\begin{example}[Theorem~\ref{thm:int-complete}]\label{ex:int-complete}
Consider $\C=\{00000, 11110, 10111, 01111, 10110, 01110, 00111, 00110\}$.  This code has
$$\CF^1(J_{\C})=\{x_1x_2x_5\} \hspace{.03in}\textrm{ and } \hspace{.03in} \CF^2(J_{\C})=\{x_i(1-x_j)~|~i \in [5];~j = 3,4;~i \neq j\}.$$
We immediately see that all elements of $\CF(J_{\C})$ satisfy $|\tau|\leq1$, and so the signature from Theorem~\ref{thm:int-complete} applies.  Thus, $\C$ is $\cap$-complete.
\end{example}

\section{Detecting local obstructions}

The primary method for showing that a code is not convex is to show that it has a local obstruction.  Section~\ref{sec:local-obs} defines local obstructions and connects them to links of certain restricted simplicial complexes.  Section~\ref{sec:non-convex} shows how to detect certain classes of local obstructions via $J_\C$ and $\CF(J_\C)$ and provides the proof of Theorem~\ref{thm:A-sigs-non-convex}.

\subsection{Local obstructions}\label{sec:local-obs}
Recall that the code in Example~\ref{ex:code-local-obs} failed to have a convex realization because the receptive field $U_1$ was covered by a pair of disjoint nonempty open sets $U_1 \cap U_2$ and $U_1 \cap U_3$, and thus no realization of $\C$ could have $U_1$ as a convex set.  In this case, the restricted cover of $U_1$ by $U_1 \cap U_2$ and $U_1 \cap U_3$ had a \emph{nerve} that was disconnected and thus, if $U_1$ were convex, there would be a topological mismatch between $U_1$ and the nerve of its restricted cover.  This topological mismatch is an example of a \emph{local obstruction}.  Specifically, the Nerve Lemma \cite[Corollary 4G.3]{Hatcher} guarantees that if $\U$ is a convex open cover (and thus a ``good cover"), then $U_\sigma$ must have the same homotopy type as $\N(\{U_\sigma \cap U_i\}_{i \in \tau})$ whenever $U_\sigma$ is non-empty and covered by a union of sets $\bigcup_{i \in \tau} U_i$, i.e.\ whenever $(\sigma, \tau) \in \RF(\C)$.  In particular, since $U_\sigma$ is the intersection of convex sets, it must be convex and hence \emph{contractible}\footnote{A set is {\it contractible} if it is homotopy-equivalent to a point, and every convex set is contractible \cite{Hatcher}.}, and thus $\N(\{U_\sigma \cap U_i\}_{i \in \tau})$ must also be contractible.  Thus, if the nerve of such a restricted cover is \underline{not} contractible, then a local obstruction is present. This restricted nerve has an alternative combinatorial formulation; specifically, 
$$\N(\{U_\sigma \cap U_i\}_{i \in \tau}) = \Lk_\sigma(\Delta|_{\sigma \cup \tau}),$$
where $\Delta|_{\sigma \cup \tau}$ is the \emph{restricted simplicial complex} 
$$\Delta|_{\sigma \cup \tau}\od \{\omega \in \Delta \mid \omega \subseteq \sigma \cup \tau\}$$
and the \emph{link} $\Lk_\sigma(\Delta|_{\sigma \cup \tau})$ is given by 
$$\Lk_\sigma(\Delta|_{\sigma \cup \tau}) = \{ \omega \in \Delta|_{\sigma \cup \tau}\mid \sigma\cap\omega = \emptyset \text{ and }\sigma \cup \omega \in \Delta|_{\sigma \cup \tau}\}.$$
This alternative characterization of the nerve yields the following formal definition of local obstruction.  For more details about local obstructions, see \cite[Section 3]{MRC}. 

\begin{definition}
Let $\C$ be a code on $n$ neurons with simplicial complex $\Delta$.  \\
For $\sigma, \tau \subseteq [n]$ with $\tau \neq \emptyset$, we say that $(\sigma, \tau)$ is a \emph{local obstruction} of $\C$ if $(\sigma, \tau) \in \RF(\C)$ and the link $\Lk_\sigma(\Delta|_{\sigma \cup \tau})$ is not contractible.
\end{definition}

As an immediate consequence of the Nerve Lemma, as described above, we obtain Lemma~\ref{lemma:loc-obs}.

\begin{lemma}\cite[Lemma 1.3]{MRC}\label{lemma:loc-obs}
If $\C$ has a local obstruction, then $\C$ is not a convex code.
\end{lemma}

\subsection{Algebraic detection of local obstructions}\label{sec:non-convex}
In general, the presence of a pseudo-monomial $x_\sigma \prod_{i \in \tau} (1-x_i) \in J_\C$ is not sufficient to guarantee that $(\sigma, \tau)$ is a RF relationship (see Table~\ref{table:J_C-and-RF}), and thus a possible candidate for a local obstruction.  This is because we cannot guarantee that $U_\sigma \cap U_i \neq \emptyset$ for all $i \in \tau$.  However, when $x_\sigma \prod_{i \in \tau}(1-x_i) $ is minimal, i.e.\ when $x_\sigma \prod_{i \in \tau} (1-x_i) \in \CF^2(J_\C)$, these conditions are guaranteed and $(\sigma, \tau) \in \RF(\C)$.  Thus, we focus on the canonical form to algebraically detect local obstructions.

\begin{lemma}\label{lemma:CF-local-obs}
For a code $\C$, if there exists $(\sigma, \tau)$ such that $x_\sigma \prod_{i \in \tau} (1-x_i) \in \CF^2(J_\C)$ and $\Lk_\sigma(\Delta|_{\sigma \cup \tau})$ is not contractible, then $\C$ is \underline{not} convex.
\end{lemma}

With this result, we can now prove Theorem~\ref{thm:A-sigs-non-convex}.  Specifically, we prove Theorem~\ref{thm:algebra-RF}, a broader result that also characterizes relevant RF conditions corresponding to these signatures.

\begin{theorem} \label{thm:algebra-RF}
If $\C$ has any of the algebraic signatures in rows A-1, A-2, A-3, or A-4 of Table~\ref{table:A-sigs-RF}, then $\C$ is not convex.  More precisely, each algebraic signature corresponds to a RF condition (as illustrated in Figure~\ref{fig:Thm2.4}), which implies that $\C$ is not convex.

\begin{table}[!h] \label{table:main}
\begin{center}
\begin{small}

\begin{tabular}{l l c l c l}
& Algebraic signature  & & Receptive field condition & & Property of $\C$\\
 \specialrule{.125em}{.2em}{.2em} 
\multirow{2}{*}{A-1} &   $\exists\; x_\sigma(1-x_i)(1-x_j) \in \CF^2(J_\C)$  & \multirow{2}{*}{$\Rightarrow$} & $(\sigma,\{i,j\}) \in \RF(\C)$ and &   \multirow{2}{*}{$\Rightarrow$} & \multirow{2}{*}{non-convex}\\
&  s.t. $x_\sigma x_i x_j \in J_\C$ & & $U_\sigma \cap U_i \cap U_j = \emptyset$ & &\\

 \specialrule{.05em}{.2em}{.2em} 
\multirow{2}{*}{A-2} & $\exists \; x_\sigma\prod_{i \in \tau}(1-x_i) \in \CF^2(J_\C)$  &  \multirow{2}{*}{$\Rightarrow$} & $(\sigma,\tau) \in \RF(\C)$  and & \multirow{2}{*}{$\Rightarrow$} & \multirow{2}{*}{non-convex}\\
&  s.t. $G_\C(\sigma,\tau)$ is disconnected & &$G_\U(\sigma,\tau)$ is disconnected & &\\

 \specialrule{.05em}{.2em}{.2em} 
\multirow{2}{*}{A-3} &  $\exists \; x_\sigma\prod_{i \in \tau}(1-x_i) \in \CF^2(J_\C)$  & \multirow{2}{*}{$\Rightarrow$} &$(\sigma,\tau) \in \RF(\C)$  and & \multirow{2}{*}{$\Rightarrow$} & \multirow{2}{*}{non-convex}\\
& s.t. $x_\sigma x_\tau \in \CF^1(J_\C)$ & & $U_\sigma \cap U_\tau = \emptyset$ but&   &\\
& & &$U_\sigma \cap U_{\tau'} \neq \emptyset \;\; \forall \; \tau' \subsetneq \tau$ &&\\

 \specialrule{.05em}{.2em}{.2em} 
\multirow{2}{*}{A-4} &  $\exists \; x_\sigma\prod_{i \in \tau}(1-x_i) \in \CF^2(J_\C)$, &  \multirow{2}{*}{$\Rightarrow$} & $(\sigma,\tau) \in \RF(\C)$ and &\multirow{2}{*}{$\Rightarrow$} & \multirow{2}{*}{non-convex}\\
&$\exists \;  \emptyset\subseteq\tilde\sigma \subseteq \sigma$ s.t. $x_{\tilde\sigma} x_\tau \in \CF^1(J_\C)$  but& & $U_\sigma \cap U_\tau \subseteq U_{\tilde\sigma} \cap U_\tau= \emptyset$&   &\\
& $x_{\sigma'}x_{\tau'} \notin \CF^1(J_\C)~\forall~\sigma'\subseteq \sigma,~\tau' \subsetneq \tau$&&but $U_\sigma \cap U_{\tau'} \neq \emptyset \;\; \forall \; \tau' \subsetneq \tau$ &&\\

 \specialrule{.125em}{.2em}{.2em} 
\end{tabular}
\caption{Algebraic signatures and receptive field conditions for non-convex codes. $G_{\C}(\sigma,\tau)$ is the simple graph on vertex set $\tau$ with edge set $\{ (ij) \in \tau \times \tau \mid x_\sigma x_i x_j \notin J_{\C} \}$.  The graph $G_\U(\sigma, \tau)$ has vertex set $\tau$ and edge set $\{(ij)\in \tau \times \tau \mid U_\sigma \cap (U_i\cap U_j) \neq \emptyset\}$.  }
\vspace{-.2in}
\label{table:A-sigs-RF}
\end{small}
\end{center}
\end{table}

\end{theorem}

\begin{figure}[!ht]
\hspace{-.2in}\includegraphics[width=1.05\textwidth]{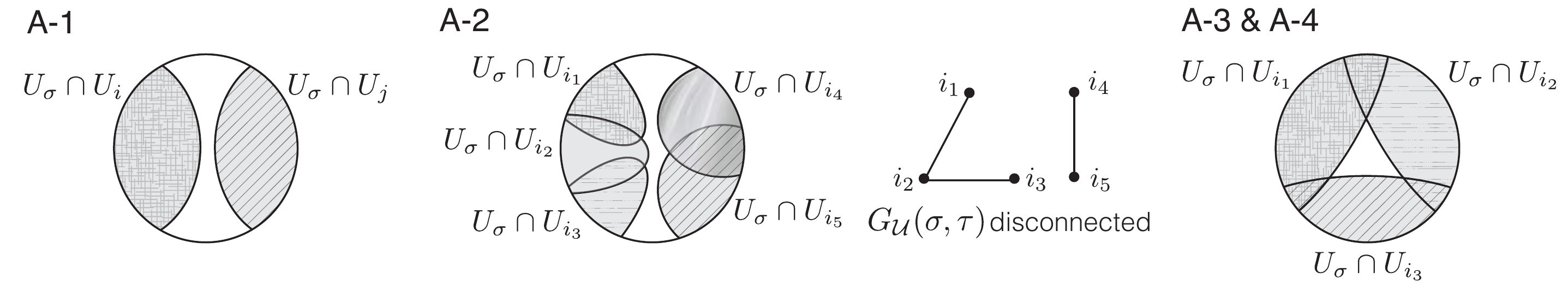}
\vspace{-.3in}
\caption{Illustrations of the RF conditions implied by signatures A-1 through A-4 in Theorem~\ref{thm:algebra-RF} (see Table~\ref{table:A-sigs-RF}).  In each picture, $U_\sigma$ is the union of the shaded regions; thus $U_\sigma$ is not contractible and hence not convex.  
}
\label{fig:Thm2.4}
\end{figure}

\FloatBarrier

\begin{proof}[Proof of Theorem~\ref{thm:algebra-RF}]
(A-1) By Lemma \ref{lemma:minRF}, $x_\sigma(1-x_i)(1-x_j) \in \CF^2(J_\C)$ implies $(\sigma,\{i,j\}) \in \RF(\C)$, and thus $U_\sigma \subseteq U_i \cup U_j$ and both $U_\sigma \cap U_i$ and $U_\sigma \cap U_j$ are non-empty.  Recall from Table~\ref{table:J_C-and-RF} that $x_\sigma x_i x_j \in J_\C$ implies $U_\sigma \cap U_i \cap U_j = \emptyset$.  Thus, $U_\sigma$ is the disjoint union of non-empty open sets $U_\sigma \cap U_i$ and $U_\sigma \cap U_j$, and so $U_\sigma$ is disconnected.  Thus, $U_\sigma = \bigcap_{k \in \sigma} U_k$ is not convex, and so some $U_k$ is not convex.  Hence $\C$ is not convex.

(A-2) By Table \ref{table:J_C-and-RF} if $\C = \C(\U)$, then $G_\C(\sigma,\tau) = G_\U(\sigma,\tau)$ since $x_\sigma x_i x_j \notin J_\C$ precisely when $U_\sigma \cap U_i \cap U_j \neq \emptyset.$  Furthermore, this graph is precisely the $1$-skeleton\footnote{The 1-skeleton of a simplicial complex is the subcomplex consisting of all faces of dimension at most 1, i.e.\ the vertices and edges of the simplicial complex; thus the 1-skeleton is the underlying graph of the simplicial complex (see e.g.\ \cite{Hatcher}).}  of $\Lk_\sigma(\Delta|_{\sigma \cup \tau})$.  Since we assume this is disconnected, it follows that $\Lk_\sigma(\Delta|_{\sigma \cup \tau})$ is not contractible, and hence $\C$ is non-convex by Lemma~\ref{lemma:CF-local-obs}.  Alternatively, $G_\U(\sigma,\tau)$ disconnected implies that $U_\sigma$ is disconnected, and hence $\C$ cannot be convex.  

(A-3) The signature for A-3 is a special case of that for A-4 since $x_\sigma x_\tau \in \CF^1(J_\C)$ guarantees $x_{\sigma'}x_{\tau'} \notin \CF^1(J_\C)$ for all $\sigma'\subseteq \sigma,~\tau' \subsetneq \tau$ by minimality of the elements in the canonical form.  Thus, we prove non-convexity of these codes via the following proof of A-4.

(A-4) Note that $x_{\tilde\sigma} x_\tau \in \CF^1(J_\C)$ implies that $U_{\tilde\sigma} \cap U_\tau= \emptyset$ by Table~\ref{table:J_C-and-RF} and thus $U_\sigma \cap U_\tau = \emptyset$ as well.  Thus, $\sigma \cup \tau \notin \Delta(\C)$ and so $\tau \notin \Lk_\sigma(\Delta|_{\sigma \cup \tau})$.  For every $\tilde\tau \subsetneq \tau$, we have $x_\sigma x_{\tilde\tau} \notin J_\C$, since if it were in $J_\C$ then some factor of it must be in $\CF^1(J_\C)$, but $x_{\sigma'}x_{\tau'} \notin \CF^1(J_\C)$ for every $\sigma'\subseteq \sigma$ and $\tau' \subseteq \tilde\tau$.  Thus, for all $\tilde\tau\subsetneq \tau$, we have $\sigma \cup \tilde\tau \in \Delta(\C)$ and so $\tilde\tau \in \Lk_\sigma(\Delta|_{\sigma \cup \tau})$; equivalently $U_\sigma \cap U_{\tilde \tau} \neq \emptyset$ for all $\tilde\tau\subsetneq \tau$.  This means $\Lk_\sigma(\Delta|_{\sigma \cup \tau})$ is a simplex missing only the top dimensional face $\tau$ (i.e.\ a hollow simplex), and so is homotopy-equivalent to a sphere, and thus is not contractible.  At the level of RF relationships, this implies that $U_\sigma$ is not contractible since it must contain a hole.  Thus, $\C$ is non-convex.

\end{proof}

As the proof of Theorem~\ref{thm:algebra-RF} illustrates, signature A-1 captures cases where $U_\sigma$ is disconnected by a pair of sets.  Signature A-2 generalizes A-1 and detects all cases where $U_\sigma$ is minimally covered by a collection of sets $U_\sigma \cap U_i$ for $i \in \tau$ in a way that forces $U_\sigma$ to be disconnected.  Note that A-2 is signature (i) from Theorem~\ref{thm:A-sigs-non-convex} in the main results (Section~\ref{sec:main-results-nonconvex}).  

Signature A-3 captures a particular case when $U_\sigma$ is minimally covered by a collection of sets $U_\sigma \cap U_i$ for $i \in \tau$ and $\Lk_\sigma(\Delta|_{\sigma \cup \tau}) = \N(\{U_\sigma \cap U_i\}_{i \in \tau})$ is a hollow simplex.  Specifically, in the case of A-3, $U_\sigma \cap U_\tau$ is the \underline{minimal} missing intersection in that for all $\tilde\sigma \subsetneq \sigma$, we have $U_{\tilde\sigma} \cap U_\tau \neq \emptyset$; thus everywhere outside of $U_\sigma$, $U_\tau$ has a non-empty intersection with each subcollection of sets from $\sigma$.  More generally, signature A-4 captures \underline{all} cases when $\Lk_\sigma(\Delta|_{\sigma \cup \tau})$ is a hollow simplex.  Specifically, the signature for A-4 does not require the minimality of the empty intersection $U_\sigma \cap U_\tau$, and so there may be a $\tilde\sigma \subsetneq \sigma$ such that $U_{\tilde\sigma} \cap U_\tau = \emptyset$, in particular we may have $U_\tau = \emptyset$.  All that is required is that every intersection of $U_\sigma$ with each proper subcollection of sets in $\tau$ is non-empty, which is guaranteed by ensuring that $U_{\sigma'} \cap U_{\tau'} \neq \emptyset$, for all $\sigma' \subseteq \sigma$ and $\tau' \subsetneq \tau$.  Signature (ii) from Theorem~\ref{thm:A-sigs-non-convex} is A-3, a special case of A-4, and so the proof of Theorem~\ref{thm:algebra-RF} completes the proof of Theorem~\ref{thm:A-sigs-non-convex}.

Note that although checking signatures A-1 and  A-2 requires determining the absence of pseudo-monomials from $J_\C$, these conditions can actually be checked from $\CF^1(J_\C)$ alone as Lemma~\ref{lemma:graph} shows.  \emph{Thus, all the algebraic signatures in Theorem~\ref{thm:A-sigs-non-convex} can be checked directly via $\CF(J_\C)$.}\\

\begin{lemma}\label{lemma:graph}
Suppose $x_\sigma \prod_{k \in \tau} (1-x_k) \in \CF^2(J_\C)$.  Then for any $i,j \in \tau$ with $i \neq j$, 
$$x_\sigma x_i x_j \in J_\C\;\; \Leftrightarrow \;\; x_{\sigma'} x_i x_j \in \CF^1(J_\C) \;\; \text{for some}
\;\; \emptyset \subseteq \sigma' \subseteq \sigma.$$
\end{lemma}

\begin{proof}
The backward direction ($\Leftarrow$) is immediate since $\CF(J_\C) \subseteq J_\C$. To see the forward direction $(\Rightarrow)$, suppose $x_\sigma x_i x_j \in J_\C.$  Then it is a multiple of some monomial $x_\omega \in \CF^1(J_\C)$ with $\omega \subseteq \sigma \cup \{i,j\}.$ There are four possibilities:\\

\noindent (1) $\omega \subseteq \sigma$.  Then $x_\omega$ divides $x_\sigma \prod_{k \in \tau} (1-x_k)$, contradicting the minimality of $x_\sigma \prod_{k \in \tau} (1-x_k)  \in \CF^2(J_\C).$\\

\noindent (2) $\omega = \sigma'\cup \{i\}$  for some $\sigma'\subseteq \sigma$.  Then $x_\sigma x_i \prod_{k \in \tau \setminus \{i\}} (1-x_k) \in J_\C,$ since it is a multiple of $x_\omega$, and hence
$$x_\sigma x_i \prod_{k \in \tau \setminus \{i\}} (1-x_k)  + x_\sigma \prod_{k \in \tau} (1-x_k) = 
x_\sigma \prod_{k \in \tau \setminus \{i\}} (1-x_k)  \in J_\C,$$
contradicting the minimality of $x_\sigma \prod_{k \in \tau} (1-x_k)  \in \CF^2(J_\C).$\\

\noindent (3) $\omega = \sigma'\cup \{j\}$ for some $\sigma'\subseteq \sigma$.  This argument is identical to the previous, leading to a contradiction.\\

\noindent Thus, the only viable possibility is:\\
\vspace{-.075in}

\noindent (4) $\omega = \sigma'\cup \{i,j\}$ for some $\sigma'\subseteq \sigma$, and thus $x_\omega = x_{\sigma'}x_ix_j\in \CF^1(J_\C)$.
\end{proof}

As mentioned at the beginning of this section, the algebraic signatures in Theorem~\ref{thm:A-sigs-non-convex} only consider minimal RF relationships as detectable by the canonical form.  The motivation for this is that other pseudo-monomials $x_\sigma \prod_{i \in \tau} (1-x_i)$ in the full ideal $J_\C$ are not guaranteed to correspond to RF relationships as $U_\sigma \cap U_i$ is not necessarily non-empty for all $i \in \tau$.  This begs the question of whether it is sufficient to only consider these minimal pseudo-monomials and minimal RF relationships.  Specifically, if for every pseudomonomial in $\CF^2(J_\C)$, we find that the links $\Lk_\sigma(\Delta|_{\sigma \cup \tau})$ are all contractible, does that guarantee that $\C$ has no local obstructions?  

Unfortunately, this is not the case.  Example~\ref{ex:non-CF-detectable} shows that there exist codes that have no local obstructions arising from pairs $(\sigma, \tau)$ with $x_\sigma \prod_{i \in \tau} (1-x_i) \in \CF^2(J_\C)$, and yet the codes still have local obstructions.

\begin{example} \label{ex:non-CF-detectable}
Consider the code $\C = \{0000, 1110, 1101, 1011, 0111\}$, where $\Delta=\Delta(\C)$ is the hollow simplex on four vertices, missing only the top-dimensional face (see Figure~\ref{fig:non-CF-detectable}A).  The canonical form is 
$$\CF(J_\C) = \{ x_1x_2x_3x_4\} \cup \{x_i(1-x_j)(1-x_k)\mid i,j,k \in [4] \text{ with } i \neq j \neq k\}.$$
The minimal RF relationships $(\sigma,\tau) = (\{i\},\{j,k\})$ detected by the canonical form all have corresponding links $\Lk_{i}(\Delta|_{\{i,j,k\}})$ that are equivalent to the simplex shown in Figure~\ref{fig:non-CF-detectable}B, and hence are contractible.  However, $\C$ has multiple local obstructions, and thus is not convex.  For example, observe that $(\{1,2\}, \{3,4\}) \in \RF(\C)$ since $(U_1 \cap U_2) \subseteq U_3 \cup U_4$ and $(U_1 \cap U_2) \cap U_i \neq \emptyset$ for $i=3, 4$, and $\Lk_{12}(\Delta|_{[4]})$ is the non-contractible disconnected graph in Figure~\ref{fig:non-CF-detectable}C.  Thus, $(\{1,2\}, \{3,4\})$ is a local obstruction, and so $\C$ cannot be convex.  Note that since $(\{1,2\}, \{3,4\})$ is a non-minimal RF relationship, its corresponding pseudo-monomial $x_1x_2(1-x_3)(1-x_4)$ is only in $J_\C$ and not in $\CF(J_\C)$.  

Similarly, $(\{1\}, \{2,3,4\})$ gives another local obstruction that is not detectable from the canonical form.  Specifically, $(\{1\}, \{2,3,4\}) \in \RF(\C)$ since $U_1 \subseteq U_2 \cup U_3 \cup U_4$ with $U_1 \cap U_i \neq \emptyset$ for each $i=2, 3, 4$, and $\Lk_{1}(\Delta|_{[4]})$ is the non-contractible hollow simplex shown in Figure~\ref{fig:non-CF-detectable}D.  In fact, it turns out that every non-maximal $\sigma \in \Delta$ has a related RF relationship that is a local obstruction (see \cite[Table 2 in Supplementary Text S1]{MRC}).

\begin{figure}[!ht]
\begin{centering}
\includegraphics[height=1.75in]{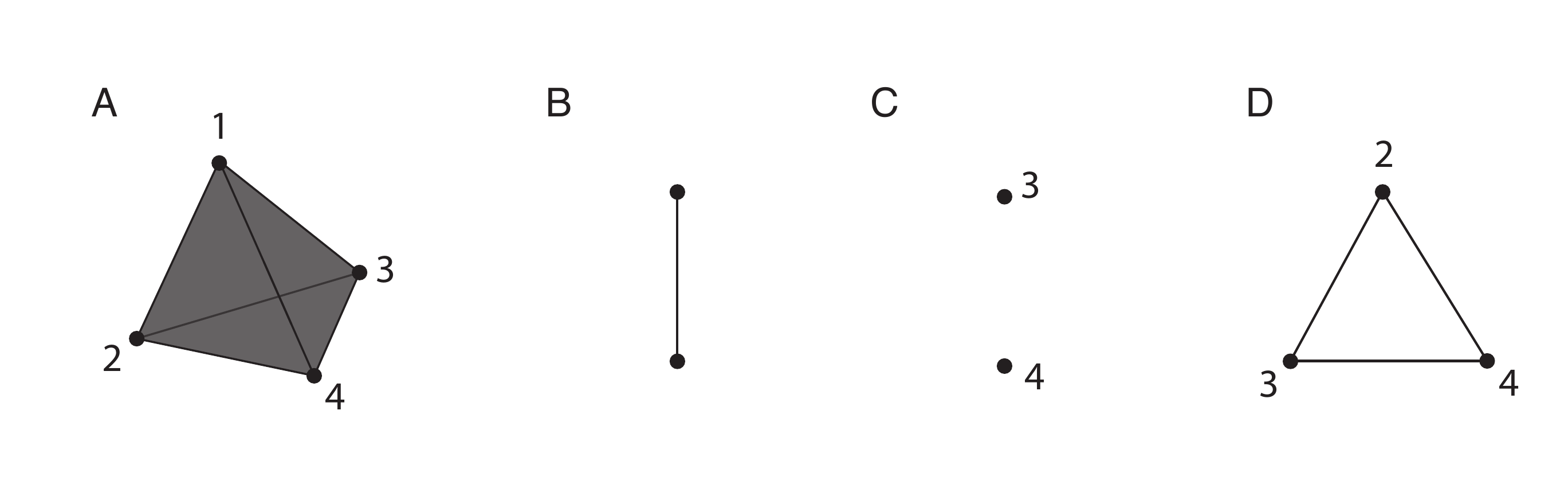}
\vspace{-.3in}
\caption{Simplicial complexes in Example~\ref{ex:non-CF-detectable}.  Note that the simplicial complex in (A) is missing the top-dimensional face $\{1,2,3,4\}$ and thus is a hollow simplex.}
\label{fig:non-CF-detectable}
\end{centering}
\end{figure} 
\end{example}

\section{Detecting convex codes}\label{sec:detecting-convex}
Recall that to prove a code is convex, it is not sufficient to show that it has no local obstructions \cite{counterexample}.  Currently the only known method for proving convexity is to construct a convex realization or show that the code belongs to a combinatorial family of codes that have been proven to be convex.  The broadest such family of codes is \emph{max $\cap$-complete} codes, for which every intersection of facets of $\Delta(\C)$ is a codeword in $\C$.  Currently, however, there is no efficient way to determine if a code is max $\cap$-complete.  Thus, we instead provide algebraic signatures of four combinatorial families of codes that all happen to be max $\cap$-complete, and thus are guaranteed to be convex by Theorem~\ref{thm:max-int-complete}. Moreover, these signatures can be checked efficiently.  Table~\ref{table:A-sigs-convex} summarizes these signatures together with the combinatorial property implied by each signature.  Section~\ref{sec:main-results-convex} provided sketches of proofs for signatures B-1 and B-2.  In this section, we prove B-3 and B-4.

\begin{table}[!ht]
\begin{center}
\begin{small}
\begin{tabular}{l l c l}
& Algebraic signature of $J_\C$ & & Property of $\C$\\
 \specialrule{.125em}{.2em}{.2em} 
 B-1 &  $\CF^1(J_\C) = \emptyset$ & $\Rightarrow$ & convex ($11\cdots1 \in \C$) \\
 \specialrule{.05em}{.2em}{.2em} 
 B-2 & $\CF^2(J_\C) = \emptyset$ & $\Rightarrow$ & convex ($\C=\Delta(\C)$)\\
  \specialrule{.05em}{.2em}{.2em} 
\multirow{2}{*}{B-3} & $\forall \; x_\sigma \in  \CF^1(J_\C)$, $|\sigma|= 2$, and &\multirow{2}{*}{$\Rightarrow$}  & convex\\
& if $x_ix_j \in \CF^1(J_\C)$, then $x_ix_k$ or $x_jx_k \in \CF^1(J_\C)$ for all $k \in [n]$&& ($\Delta(\C)$ has disjoint facets)\\
  \specialrule{.05em}{.2em}{.2em} 
B-4 & $\forall \; x_\sigma\prod_{i \in \tau}(1-x_i) \in \CF^2(J_\C),$ $|\tau| = 1$ & $\Rightarrow$ & convex ($\C$ is $\cap$--complete)\\
 \specialrule{.125em}{.2em}{.2em} 
\end{tabular}
\caption{Algebraic signatures of convex codes.}
\label{table:A-sigs-convex}
\end{small}
\end{center}
\vspace{-.15in}
\end{table}

\paragraph{Proof of B-3.}
Recall that signature B-1 captures when a code contain the all-ones word, and thus the corresponding simplicial complex has a single facet.  As a generalization, we consider codes whose simplicial complexes have disjoint facets, which are also provably convex \cite{MRC}.  In the following, we show these codes can be algebraically detected via the signature B-3, but first we need the following lemma.  

\begin{lemma}\label{lemma:disjoint-cliques}
The graph, or 1-skeleton, of $\Delta(\C)$ is a disjoint union of maximal cliques if and only if the following property holds:\\

\vspace{-.15in}
\hspace{1in} if $x_ix_j \in \CF^1(J_\C)$, then $x_ix_k$ or $x_jx_k \in \CF^1(J_\C)$ for all $k \in [n]$. \hspace{1in} $(*)$ 
\end{lemma}

\begin{proof}
Let $G$ be the underlying graph of $\Delta=\Delta(\C)$, i.e.\ its 1-skeleton.  Observe that $x_ix_j \in \CF^1(J_\C)$ precisely when $\{i,j\}$ is a minimal set missing from $\Delta$, and so $i$ and $j$ are vertices of $G$, but the edge $(ij)$ is missing from $G$. 
 
\noindent ($\Rightarrow$)  If $G$ is a disjoint union of maximal cliques, then whenever two vertices $i$ and $j$ are in distinct maximal cliques, no other vertex $k$ can be adjacent to both $i$ and $j$.  This means that whenever $x_ix_j \in \CF^1(J_\C)$, for every $k\in [n]$, at least one of $x_ix_k$ or $x_jx_k \in J_\C$.  Since these elements are minimal, we must have $x_ix_k$ or $x_jx_k \in \CF^1(J_\C)$ (because $x_k \notin J_\C$ for any $k$ since we assume $U_k \neq \emptyset$).

\noindent($\Leftarrow$)  We prove this by contrapositive.  Suppose that $G$ is not the disjoint union of maximal cliques.  Then there exist distinct vertices $i,j,k \in G$ that are connected, but do not form a clique; specifically, $(ik)$ and $(jk)$ are edges in $G$, but $(ij)$ is not.  Since $(ij) \notin G$, $x_ix_j \in \CF^1(J_\C)$, but neither $x_ix_k$ nor $x_jx_k$ is in $\CF^1(J_\C)$; thus violating the condition on $\CF^1(J_\C)$ from the statement.    
\end{proof}

Satisfying property (*) from Lemma~\ref{lemma:disjoint-cliques} alone is not sufficient to guarantee convexity, as the following example shows.
\begin{example}
Consider $\C = \{000000, 111000, 110100, 101100, 000011, 110000, 101000, 100100\}$ with 
$$\CF^1(J_\C) = \{x_1x_5,~x_1x_6,~x_2x_5,~x_2x_6,~x_3x_5,~x_3x_6,~x_4x_5,~x_4x_6,~x_2x_3x_4\} \text{ and} $$ 
$$\CF^2(J_\C) = \{x_1(1-x_2)(1-x_3)(1-x_4),~x_2(1-x_1),~x_3(1-x_1),~x_4(1-x_1),~x_5(1-x_6),~x_6(1-x_5)\}.$$
Observe that property $(*)$ from Lemma~\ref{lemma:disjoint-cliques} holds for $\CF^1(J_\C)$, and the graph of $\Delta(\C)$ is the disjoint union of the 4-clique on $\{1,2,3,4\}$ and the edge $(56)$.  $\C$ is not convex, however, because it satisfies signature A-4 via $x_1(1-x_2)(1-x_3)(1-x_4) \in \CF^2(J_\C)$ and $x_2x_3x_4 \in \CF^1(J_\C)$ which forces a local obstruction since $\Lk_1(\Delta|_{[4]})$ is the non-contractible hollow triangle.  
\end{example}

Despite not necessarily being convex, codes satisfying property $(*)$ from Lemma~\ref{lemma:disjoint-cliques} display an interesting relation.  Consider the \emph{co-firing relation} defined by $i \sim j$ if and only if neurons $i$ and $j$ ``co-fire" together in $\C$, i.e.\ $\{i,j\} \subseteq \sigma$ for some $\sigma \in \C$.  It is straightforward to check that this is an equivalence relation precisely when the code satisfies property $(*)$, as that condition on the 1-skeleton of $\Delta(\C)$ ensures transitivity. 

We now turn to the subclass of codes described in Lemma~\ref{lemma:disjoint-cliques} that have simplicial complexes with disjoint facets, and thus are guaranteed to be convex \cite{MRC}.  This will complete the proof of B-3.

\begin{proposition}\label{prop:disjoint-facets}
Given a code $\C$, its simplicial complex $\Delta(\C)$ has disjoint facets if and only if the following two properties hold:
\begin{enumerate}
\item For all $x_\sigma \in \CF^1(J_\C)$, we have $|\sigma| = 2$; and
\item If $x_ix_j \in \CF^1(J_\C)$, then $x_ix_k$ or $x_jx_k \in \CF^1(J_\C)$ for all $k \in [n]$.
\end{enumerate}
\end{proposition}
\begin{proof}
Observe that $\Delta=\Delta(\C)$ has disjoint facets if and only if $\Delta$ is a disjoint union of simplices. This occurs precisely when (a) the 1-skeleton of $\Delta$ is a disjoint union of maximal cliques and (b) each maximal clique is in $\Delta$ (i.e.\ each  maximal clique yields a simplex in $\Delta$).  

By Lemma~\ref{lemma:disjoint-cliques}, (a) holds if and only if Property 2 is satisfied.  Note that (b) holds if and only if every clique of the 1-skeleton is in $\Delta$, not just the maximal cliques, since $\Delta$ is closed under taking subsets.  Property 1 guarantees that $x_\omega \in J_\C$ (i.e.\ $\omega \notin \Delta$) if and only if $x_\omega$ is a multiple of $x_ix_j \in \CF^1(J_\C)$ for some $i, j \in \omega$, and so $\{i,j\} \notin \Delta$.  Thus Property 1 ensures $\omega \notin \Delta$ if and only if $\omega$ is missing an edge $(ij)$, and thus is \underline{not} a clique.  Hence (b) holds if and only if Property 1 is satisfied.
\end{proof}

Note that the signature in B-3 relies solely on $\CF^1(J_\C)$, which is a generating set for the Stanley-Reisner ideal of $\Delta(\C)$; thus, this property can be read off from the Stanley-Reisner ideal alone.

\paragraph{Proof of B-4.}  Recall from B-2 that a code $\C$ is a simplicial complex precisely when all pseudo-monomials in $\CF(J_\C)$ have $|\tau|=0$, and in this case $\C$ is provably convex, and consequently has no local obstructions.  It turns out that when all pseudo-monomials in $\CF^2(J_\C)$ have $|\tau|=1$, Corollary~\ref{cor:no-local-obs} shows that no local obstructions can arise involving the neurons in $\sigma$ and $\{i\}$.  To prove this, we first need the following two lemmas.

\begin{lemma}\label{lemma:J_C-algebra}
Let $\rho \in \F_2[x_1, \ldots, x_n]$ and $i \in [n]$.  If $\rho \notin J_\C$ but $\rho (1-x_i) \in J_\C$, then $\rho x_i \notin J_\C$.  
\end{lemma}
\begin{proof}
If $\rho x_i \in J_\C$, then the sum $\rho(1-x_i) + \rho x_i \in J_\C$, but $\rho(1-x_i) + \rho x_i = \rho \notin J_\C$ by hypothesis.  Thus, $\rho x_i \notin J_\C$.
\end{proof}

\begin{lemma}\label{lemma:link-contractible}
For a code $\C$, if $x_\sigma \notin J_\C$ but $x_\sigma(1-x_i) \in J_\C$, then for any $\tilde\sigma \supseteq \sigma$ and any $\tau \supseteq \{i\}$, the link $\Lk_{\tilde\sigma}(\Delta|_{\tilde\sigma \cup \tau})$ is contractible. 
\end{lemma}
\begin{proof}
To show the link $\Lk_{\tilde\sigma}(\Delta|_{\tilde\sigma \cup \tau})$ is contractible, we show it is a cone.  Let $\tau' \in \Lk_{\tilde\sigma}(\Delta|_{\tilde\sigma \cup \tau})$; note $\tau' \subseteq \tau$.  This implies $\tilde\sigma\cup \tau' \in (\Delta|_{\tilde\sigma \cup \tau}) \subseteq \Delta(\C)$, and so $\rho = x_{\tilde\sigma}x_{\tau'} \notin J_\C$.  But $\rho(1-x_i) \in J_\C$ since it is a multiple of $x_\sigma(1-x_i) \in J_\C$.  Thus by Lemma~\ref{lemma:J_C-algebra}, $\rho x_i \notin J_\C$ which implies 
$\tilde\sigma\cup \tau' \cup \{i\} \in \Delta(\C)$, and so $\tau' \cup \{i\} \in \Lk_{\tilde\sigma}(\Delta|_{\tilde\sigma \cup \tau})$.  Thus, $i$ is a cone point of $\Lk_{\tilde\sigma}(\Delta|_{\tilde\sigma \cup \tau})$, and so the link is contractible.  
\end{proof}

\begin{corollary}\label{cor:no-local-obs}
For a code $\C$, if $x_\sigma(1-x_i) \in \CF(J_\C)$, then $(\tilde\sigma, \tau)$ is not a local obstruction of $\C$ for any $\tilde\sigma \supseteq \sigma$ and any $\tau \supseteq \{i\}$.
\end{corollary}
\begin{proof}
Observe that if $x_\sigma(1-x_i) \in \CF(J_\C)$, then by minimality $x_\sigma \notin J_\C$.  Thus by Lemma~\ref{lemma:link-contractible}, for any $(\tilde\sigma, \tau)$, the link $\Lk_{\tilde\sigma}(\Delta|_{\tilde\sigma \cup \tau})$ is contractible, and thus $(\tilde\sigma, \tau)$ cannot be a local obstruction.
\end{proof}

Corollary~\ref{cor:no-local-obs} shows that no RF relationship containing a minimal RF relationship $(\sigma, \{i\})$ (i.e.\ when $x_\sigma(1-x_i) \in \CF(J_\C)$) can produce a local obstruction.  Thus, if a code only has minimal RF relationships of the form $(\sigma, \emptyset)$ or $(\sigma, \{i\})$, then it has \underline{no} local obstructions, since every RF relationship $(\sigma, \tau)$ must contain one of these minimal RF relationships.  Such a code can be immediately identified from its canonical form, as every pseudo-monomial will satisfy $|\tau|\leq 1$.  Furthermore, Proposition~\ref{prop:int-complete} shows that these codes are precisely $\cap$-complete codes.  Since $\cap$-complete codes are max $\cap$-complete, Theorem~\ref{thm:max-int-complete} guarantees these codes are in fact convex beyond simply having no local obstructions.  This completes the proof of B-4 and Theorem~\ref{thm:int-complete} from Section~\ref{sec:main-results-convex}.

\begin{proposition}\label{prop:int-complete}
For a code $\C$, the following are equivalent:\
\begin{enumerate}[(1)]
\item Every pseudo-monomial $x_\sigma \prod_{i \in \tau} (1-x_i)$ in $\CF(J_\C)$ has $|\tau| \leq 1$, 
\item For each RF relationship $(\sigma, \tau)$ with $\tau \neq \emptyset$, there exists an $i \in \tau$ such that $(\sigma, \{i\})$ is also an RF relationship, and
\item  $\C$ is $\cap$-complete.
\end{enumerate}
\end{proposition}

In the proof, we will use the notation $\U|_{\sigma \cup \tau}$ to denote the subcover $\U|_{\sigma \cup \tau}\od \{U_i \mid i \in \sigma \cup \tau\}$.  Also we use $\C|_{\sigma \cup \tau}$ to denote the restricted code $\C|_{\sigma \cup \tau} \od \{ \omega \in \C \mid \omega \subseteq \sigma \cup \tau\}$.  Note that if $\C = \C(\U)$ then $\C|_{\sigma \cup \tau} = \C(\U|_{\sigma \cup \tau})$.  We also use the straightforward fact that if $\C$ is $\cap$-complete, then $\C|_{\sigma \cup \tau}$ is $\cap$-complete for any $\sigma \cup \tau \subset [n]$.

\begin{proof} 
We prove $(1) \Leftrightarrow (2)$, and then $(1) \Leftrightarrow (3)$.  

\noindent (1)$\Leftrightarrow$(2): By Lemma~\ref{lemma:minRF}, there is a one-to-one correspondence between pseudo-monomials in $\CF(J_\C)$ and minimal RF relationships.  Thus, every $x_\sigma \prod_{i \in \tau} (1-x_i)$ in $\CF(J_\C)$ has $|\tau| \leq 1$ if and only if the only minimal RF relationships are those of the form $(\sigma, \emptyset)$ and $(\sigma, \{i\})$ for some $i \in [n]$.  Recall that if $(\sigma, \tau)$ is a RF relationship with $\tau \neq \emptyset$, then $U_\sigma \neq \emptyset$ since $U_\sigma \cap U_i \neq \emptyset$ for all $i \in \tau$; hence $(\sigma, \tau)$ contains a minimal RF relationship $(\sigma', \tau')$ for some non-empty $\sigma' \subseteq \sigma$ and non-empty $\tau' \subseteq \tau$.  We see that (1) holds if and only if each such RF relationship $(\sigma, \tau)$ contains a minimal relationship of the form $(\sigma', \{i\})$ for some $i \in \tau$, in which case, $(\sigma, \{i\})$ is also an RF relationship since $U_\sigma \subseteq U_{\sigma'}$.  Thus (1) and (2) are equivalent.\\

\noindent (1)$\Rightarrow$(3):\  Suppose $\CF(J_\C)$ only contains pseudo-monomials with $|\tau| \leq 1.$  Consider a pair of overlapping codewords $\omega_1, \omega_2 \in \C$ and let $\sigma = \omega_1 \cap \omega_2$.   To obtain a contradiction, suppose $ \sigma \notin \C$.   Then $\sigma \in \Delta(\C) \setminus \C$, since it is a subset of $\omega_1, \omega_2 \in \Delta(\C)$.  This implies that $x_\sigma \prod_{i \in \tau} (1-x_i) \in J_\C$ for some $\tau \neq \emptyset$ with $\sigma \cap \tau = \emptyset$.  It follows that $x_\sigma \prod_{i \in \tau} (1-x_i)$ must be a multiple of a Type 2 pseudo-monomial in the canonical form, say $x_{\sigma'}(1-x_i) \in\CF(J_\C),$ where $i \in \tau$ and $\sigma' \subseteq \sigma$ is nonempty.  In particular, $x_\sigma (1-x_i) \in J_\C$, which implies $U_\sigma \subseteq U_i$ by Table~\ref{table:J_C-and-RF}.  Since $\sigma \subseteq \omega_1$, $U_{\omega_1} \subseteq U_\sigma \subseteq U_i$ which implies $i \in \omega_1$ since otherwise the region corresponding to $\omega_1$ would be covered and thus could not produce a codeword.  Similarly, $i \in \omega_2$, and so $i \in \omega_1 \cap \omega_2 = \sigma$.  But then $ i \in \sigma \cap \tau$ contradicting $\sigma \cap \tau = \emptyset$.  Hence $\sigma \in \C$ and $\C$ is $\cap$-complete.\\

\noindent (3)$\Rightarrow$(1):\ Suppose $\C$ is $\cap$-complete, and consider an element $x_\sigma \prod_{i \in \tau} (1-x_i) \in \CF(J_\C).$  Note that for any cover $\U$ with $\C=\C(\U)$, this implies $U_\sigma \subseteq \bigcup_{i \in \tau} U_i$ in $\U|_{\sigma \cup \tau}$, and so $\sigma \notin \C|_{\sigma \cup \tau}.$  To obtain a contradiction, assume $|\tau| > 1$ and let $j, k \in \tau$ with $j \neq k$.  

Suppose that $U_{\sigma \cup \{j\}} \subseteq \bigcup_{i \in \tau\setminus \{j\}} U_i$.  Then it would follow that the portion of $U_\sigma$ that is covered by $U_j$ (i.e.\ $U_\sigma \cap U_j = U_{\sigma \cup \{j\}})$ is also covered by $\bigcup_{i \in \tau \setminus \{j\}} U_i$.  This would imply that $U_\sigma \subseteq \bigcup_{i \in \tau\setminus \{j\}} U_i$, but that would contradict the minimality of $U_\sigma \subseteq \bigcup_{i \in \tau} U_i$ guaranteed by $x_\sigma \prod_{i \in \tau} (1-x_i) \in \CF(J_\C)$.  Hence, $U_{\sigma \cup \{j\}} \nsubseteq \bigcup_{i \in \tau\setminus \{j\}} U_i$, and so 
$$U_{\sigma \cup \{j\}} \setminus \bigcup_{i \in \tau\setminus \{j\}} U_i \neq \emptyset$$
ensuring that $\sigma \cup \{j\} \in \C|_{\sigma \cup \tau}.$\footnote{Alternatively, we can see $\sigma \cup \{j\} \in \C|_{\sigma \cup \tau}$ algebraically, by considering $\rho = x_\sigma \prod_{i \in \tau\setminus j} (1-x_i)$.  We have $\rho (1-x_j) \in \CF(J_\C)$ while $\rho \notin J_\C$, by minimality of elements in $\CF(J_\C)$, and so Lemma~\ref{lemma:J_C-algebra} guarantees $\rho x_j = x_\sigma x_j \prod_{i \in \tau\setminus j} (1-x_i) \notin J_\C$.  Thus $\sigma \cup \{j\} \in \C|_{\sigma \cup \tau}$, since the absence of that pseudo-monomial from $J_\C$ implies $U_{\sigma \cup j}$ is nonempty and is not covered sets in $\tau \setminus \{j\}$.}  By the same argument, we have $\sigma \cup \{k\} \in \C|_{\sigma \cup \tau}$.

Let $\omega_j = \sigma \cup \{j\}$ and $\omega_k = \sigma \cup \{k\}$.  Since the restricted code $\C|_{\sigma \cup \tau}$ is $\cap$-complete, this implies that $\sigma = \omega_j \cap \omega_k$ must be in $\C|_{\sigma \cup \tau}$.  But this contradicts the hypothesis that $x_\sigma \prod_{i \in \tau} (1-x_i) \in \CF(J_\C)$ with $j, k\in \tau$, which guaranteed $\sigma \notin \C|_{\sigma \cup \tau}.$  Thus, we conclude that $|\tau| \leq 1$ and (1) holds.  
\end{proof}

\section{Examples}\label{sec:examples}
In this section, we provide examples to illustrate how one can use the algebraic signatures summarized in Table~\ref{table:A-signatures} to detect convexity or non-convexity.  In Section~\ref{sec:examples-main-results}, examples of codes satisfying A-2, A-3, and B-4 were given, and so we do not give examples for those signatures here.

\begin{table}[!h]
\begin{center}
\begin{small}
\begin{tabular}{l l c l}
& Algebraic signature of $J_\C$ & & Property of $\C$\\
 \specialrule{.125em}{.2em}{.2em} 
 A-1 &  $\exists\; x_\sigma(1-x_i)(1-x_j) \in \CF^2(J_\C)$ s.t. $x_\sigma x_i x_j \in J_\C$ & $\Rightarrow$ & non-convex\\
 \specialrule{.05em}{.2em}{.2em} 
A-2 & $\exists \; x_\sigma\prod_{i \in \tau}(1-x_i) \in \CF^2(J_\C)$ s.t. $G_\C(\sigma,\tau)$ is disconnected & $\Rightarrow$ & non-convex\\
 \specialrule{.05em}{.2em}{.2em} 
A-3 &  $\exists \; x_\sigma\prod_{i \in \tau}(1-x_i) \in \CF^2(J_\C)$ s.t. $x_\sigma x_\tau \in \CF^1(J_\C)$ & $\Rightarrow$& non-convex\\
 \specialrule{.05em}{.2em}{.2em} 
\multirow{2}{*}{A-4} &  $\exists \; x_\sigma\prod_{i \in \tau}(1-x_i) \in \CF^2(J_\C)$, and $\exists \; \tilde\sigma \subseteq \sigma$ s.t. $x_{\tilde\sigma} x_\tau \in \CF^1(J_\C)$&  \multirow{2}{*}{$\Rightarrow$} & \multirow{2}{*}{non-convex}\\
& but $x_{\sigma'}x_{\tau'} \notin \CF^1(J_\C)$ for all  $\sigma'\subseteq \sigma,~\tau' \subsetneq \tau$\\
 \specialrule{.125em}{.2em}{.2em} 

\\

 \specialrule{.125em}{.2em}{.2em} 
 B-1 &  $\CF^1(J_\C) = \emptyset$ & $\Rightarrow$ & convex ($11\cdots1 \in \C$) \\
 \specialrule{.05em}{.2em}{.2em} 
 B-2 & $\CF^2(J_\C) = \emptyset$ & $\Rightarrow$ & convex ($\C=\Delta(\C)$)\\
  \specialrule{.05em}{.2em}{.2em} 
\multirow{2}{*}{B-3} & $\forall \; x_\sigma \in  \CF^1(J_\C)$, $|\sigma|=2$, and &\multirow{2}{*}{$\Rightarrow$}  & convex\\
& if $x_ix_j \in \CF^1(J_\C)$, then $x_ix_k$ or $x_jx_k \in \CF^1(J_\C)$ for all $k \in [n]$&& ($\Delta(\C)$ has disjoint facets)\\
  \specialrule{.05em}{.2em}{.2em} 
B-4 & $\forall \; x_\sigma\prod_{i \in \tau}(1-x_i) \in \CF^2(J_\C),$ $|\tau| = 1$ & $\Rightarrow$ & convex ($\C$ is $\cap$--complete)\\
 \specialrule{.125em}{.2em}{.2em} 
\end{tabular}
\caption{Algebraic signatures of convex and non-convex codes.  $G_{\C}(\sigma,\tau)$ is the simple graph on vertex set $\tau$ with edge set $\{ (ij) \in \tau \times \tau \mid x_\sigma x_i x_j \notin J_{\C} \}$.
}
\label{table:A-signatures}
\end{small}
\end{center}
\end{table}

\begin{example}[signature A-1]\label{ex:A-1}
Consider $\C_1=\{000, 110, 101, 011\}$. This code has
$$\CF^1(J_{\C_1})=\{x_1x_2x_3\} \hspace{.03in}\textrm{ and } \hspace{.03in} \CF^2(J_{\C_1})=\{ x_i(1-x_j)(1-x_k)~|~i,j,k = 1, 2, 3; \textrm{ and all indices distinct}\}.$$
Observe that $x_1(1-x_2)(1-x_3) \in \CF^2(J_{\C_1})$ and $x_1x_2x_3 \in \CF^1(J_{\C_1}) \subseteq J_{\C_1}$.  Thus signature A-1 applies, and so $\C_1$ is not convex.  

At the level of receptive fields, the obstruction is that $U_1$ is the disjoint union of nonempty sets $U_1\cap U_2$ and $U_1 \cap U_3$; hence $U_1$ is disconnected and not convex.
 \end{example}
\smallskip

\begin{example}[signature A-4]\label{ex:A-4}
Consider $\C_2=\{0,1\}^5 \setminus \{11000, 10111, 11111\}$.  Then
$$\CF^1(J_{\C_2})=\{x_1x_3x_4x_5\} \hspace{.03in}\textrm{ and } \hspace{.03in} \CF^2(J_{\C_2})=\{ x_1x_2(1-x_3)(1-x_4)(1-x_5)\}.$$
Consider $\sigma = \{1,2\}$ and $\tau = \{3, 4, 5\}$, so that $x_\sigma \prod_{i \in \tau}(1-x_i) \in \CF^2(J_{\C_2})$.  For $\tilde \sigma = \{1\}$, we see $x_{\tilde\sigma}x_\tau \in \CF^1(J_{\C_2})$ and for all  $\sigma'\subseteq \sigma,~\tau' \subsetneq \tau$, 
$x_{\sigma'}x_{\tau'} \notin \CF^1(J_{\C_2})$.  Thus A-4 applies, and so $\C_2$ is not convex.

In terms of receptive fields, we have that $U_1 \cap U_\tau = \emptyset$, and so $U_\sigma \cap U_\tau = \emptyset$.  But for all $\sigma' \subseteq \sigma$ and $\tau' \subsetneq \tau$, we have $U_{\sigma'} \cap U_{\tau'} \neq \emptyset$.  Hence, the collection of sets $U_\sigma \cap U_i$ for $i \in \tau$ form a hollow simplex covering $U_\sigma$, forcing $U_\sigma$ to contain a hole.  
\end{example}
\smallskip

\begin{example}[signature B-1]\label{ex:B-1}
Consider $\C_4=\{000, 110, 101, 111\}$.  Then 
$$\quad \CF^1(J_{\C_4})=\emptyset \hspace{.03in}\textrm{ and } \hspace{.03in} \CF^2(J_{\C_4})=\{x_2(1-x_1), \ x_3(1-x_1),\ x_1(1-x_2)(1-x_3)\}.$$
Since $\CF^1(J_{\C_4})$ is empty, signature B-1 applies, and so $\C_4$ is convex.  We could have also seen this directly from the fact that $111 \in \C_4$.
\end{example}
\smallskip

\begin{example}[signature B-2]\label{ex:B-2}
Consider $\C_5=\{0000, 1000, 0100, 0010, 0001, 1100, 1010, 0110, 0011, 1110\}$.  This code has
$$\CF^1(J_{\C_5})=\{x_1x_4, \ x_2x_4\} \hspace{.03in}\textrm{ and } \hspace{.03in}  \CF^2(J_{\C_5})=\emptyset.$$
Since $\CF^2(J_{\C_5})$ is empty, signature B-2 applies, and so $\C_5$ is a simplicial complex, and hence is convex.  
\end{example}
\smallskip

\begin{example}[signature B-3]\label{ex:B-3}
Consider $\C_6=\{0000, 0100, 0001, 1100, 1010, 1110\}$.  This code has
$$\CF^1(J_{\C_6})=\{x_1x_4, \ x_2x_4, \ x_3x_4\} \hspace{.03in}\textrm{ and } \hspace{.03in} \CF^2(J_{\C_6})=\{x_3(1-x_1), \ x_1(1-x_2)(1-x_3)\}.$$
Observe that all the elements of $\CF^1(J_{\C_6})$ have $|\sigma|= 2$, satisfying the first part of signature B-3.  For $x_1x_4 \in \CF^1(J_{\C_6})$, we have $x_2x_4 \in \CF^1(J_{\C_6})$ and $x_3x_4 \in \CF^1(J_{\C_6})$ so the second condition holds for $i=1$, $j=4$, and $k=2,3$.  It is easy to see the condition also holds for $i=2$, $j=4$ and for $i=3$, $j=4$.  Thus, signature B-3 applies, so $\Delta(\C_6)$ has disjoint facets (specifically, $\{1,2,3\}$ and $\{4\}$) and $\C_6$ is convex.  
\end{example}
\smallskip

The signatures A-1 through A-4 guarantee non-convexity by way of a local obstruction that can be detected from the canonical form, while signatures B-1 through B-4 guarantee convexity, and thus no local obstructions.  However, Example~\ref{ex:non-CF-detectable} showed that even in the absence of local obstructions corresponding to elements of $\CF(J_\C)$, a code $\C$ may still have local obstructions and thus be provably non-convex.  Additionally, even when a code has no local obstructions of any type, it may still be non-convex (see $\C_2$ in Example~\ref{ex:diff-generalizations}(b), first observed to be non-convex in \cite{counterexample}).  Despite these complicating factors, it may still be useful to identify when a code cannot have any local obstructions ``arising" from canonical form elements; more precisely, it has no \emph{$\CF$-detectable} local obstructions, as defined below.  

\begin{definition}\label{def:CF-detectable}
A local obstruction $(\sigma,\tau)$ is {\it $\CF$-detectable} if there exists a local obstruction
$(\sigma',\tau')$ with $\sigma' \subseteq \sigma$ and $\tau' \subseteq \tau,$ such that $(\sigma',\tau')$ is a minimal RF relationship.
\end{definition}

C-1 and C-2 give two algebraic signatures of codes with no $\CF$-detectable local obstructions.  Supplemental Text S2 provides more background on $\CF$-detectable local obstructions and Theorem~\ref{thm:A-sigs-no-CF-detectable} proving these signatures.

\begin{table}[!h]
\begin{center}
\begin{small}
\begin{tabular}{l l c l}
& Algebraic signature of $J_\C$ & & Property of $\C$\\
 \specialrule{.125em}{.2em}{.2em} 
\multirow{2}{*}{C-1} & \multirow{2}{*}{$\forall \; x_\sigma\prod_{i \in \tau}(1-x_i) \in \CF^2(J_\C)$, $x_\sigma x_\tau \notin J_\C$} & \multirow{2}{*}{$\Rightarrow$}  & no $\CF$-detectable\\ 
& & & local obstructions\\
 \specialrule{.05em}{.2em}{.2em} 
\multirow{2}{*}{C-2} & $\forall \; x_\sigma\prod_{i \in \tau}(1-x_i) \in \CF^2(J_\C)$,&\multirow{2}{*}{$\Rightarrow$} &  no $\CF$-detectable\\
 &  $\exists \; i\in \tau$ s.t. $x_ix_\omega \notin \CF^1(J_\C)$ for all $\omega \subseteq \sigma\cup \tau$ & & local obstructions \\ 
 \specialrule{.125em}{.2em}{.2em} 
\end{tabular}
\caption{Algebraic signatures of codes with no $\CF$-detectable local obstructions.  These codes are \underline{not} guaranteed to be convex or non-convex, but do not have any local obstructions that can be detected from the canonical form.}
\label{table:A-signatures-CF-detectable}
\end{small}
\end{center}
\vspace{-.15in}
\end{table}


\begin{example}[signature C-1]\label{ex:C-1}
Consider the code
\begin{eqnarray*}
&\C_7=&\{0000000, 0100000, 0010000, 0001000, 0000100, 0000010, 1100000,\\
&& 1010000, 1001000, 0110000, 0101000, 0011000, 0001100, 0000110, \\
&& 0000101, 0000011, 1110000, 1101000, 1011000, 0111000, 0000111, 1111000\}.
\end{eqnarray*}
More compactly, we can describe the codewords as subsets of active neurons, and we obtain 
$$\C_7=\{\emptyset, 2, 3, 4, 5, 6, 12, 13, 14, 23, 24, 34, 45, 56, 57, 67, 123, 124, 134, 234, 567, 1234\},$$
with maximal codewords $1234, 45,$ and $567$.  This code has 
\begin{eqnarray*}
&&\CF^1(J_{\C_7})=\{x_1x_5, \ x_1x_6, \ x_1x_7, \ x_2x_5, \ x_2x_6, \ x_2x_7, \ x_3x_5, \ x_3x_6, \ x_3x_7, \ x_4x_6, \ x_4x_7 \}, \;\textrm{and } \\
&& \CF^2(J_{\C_7})=\{ x_1(1-x_2)(1-x_3)(1-x_4), \ x_7(1-x_5)(1-x_6)\}.
\end{eqnarray*}

Since all the elements of $\CF^1(J_{\C_7})$ have $|\sigma|\leq 2$, we might attempt to apply signature B-3 to guarantee convexity; however, that signature fails here since $x_1x_5 \in \CF^1(J_{\C_7})$ but neither $x_1x_4 \in \CF^1(J_{\C_7})$ nor $x_4x_5 \in \CF^1(J_{\C_7})$.  Thus, we must turn to $\CF^2(J_{\C_7})$.  For $ x_1(1-x_2)(1-x_3)(1-x_4)$, we see that $x_1x_2x_3x_4 \notin J_{\C_7}$ since no factor of it is in $\CF^1(J_{\C_7})$.  Similarly, for $x_7(1-x_5)(1-x_6)$, we see that $x_5x_6x_7 \notin J_{\C_7}$ since it has no factors in $\CF^1(J_{\C_7})$.  Thus, signature C-1 is satisfied and $\C_7$ has no $\CF$-detectable local obstructions.  This signature does not enable us to conclude anything about the convexity of $\C_7$; however, $\C_7$ is in fact convex, as it is max $\cap$-complete.

\begin{figure}[!ht]
\begin{centering}
\includegraphics[height=1in]{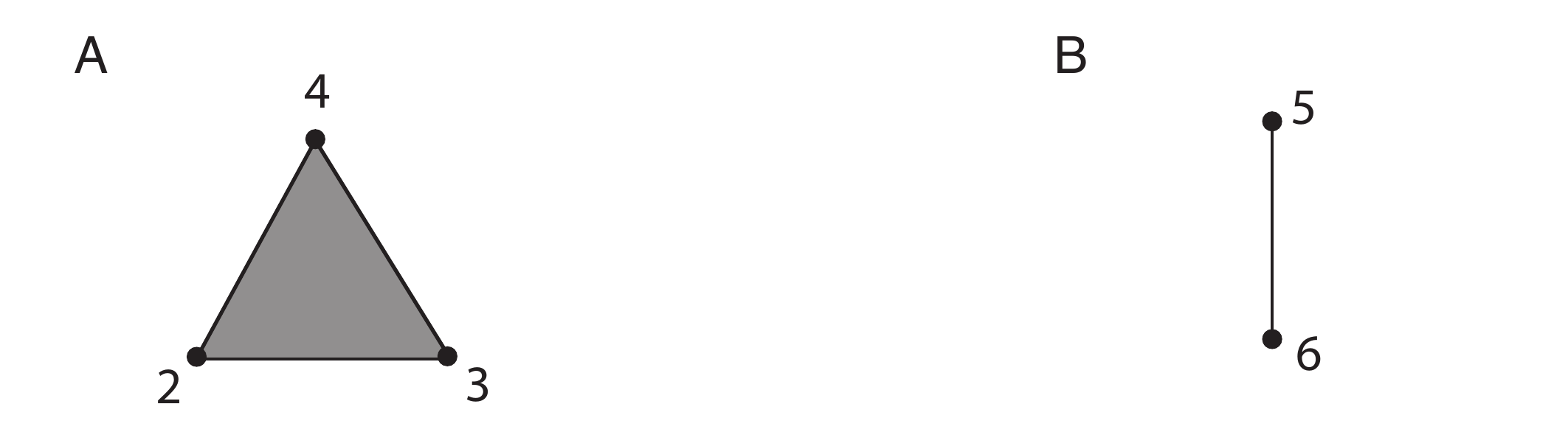}
\caption{Links from Example~\ref{ex:C-1}.  (A) $\Lk_1(\Delta|_{1234})$ is a simplex.  (B) $\Lk_7(\Delta|_{567})$ is a simplex.}
\label{fig:simplex-links}
\end{centering}
\end{figure} 
\end{example}
\smallskip

Recall from Section~\ref{sec:local-obs} that the source of $\CF$-detectable local obstructions is non-contractible links of the form $\Lk_\sigma(\Delta|_{\sigma \cup \tau})$, where $x_\sigma \prod_{i \in \tau}(1-x_i) \in \CF^2(J_\C)$.  For $\C_7$, the relevant links occur when $\sigma=\{1\}$ and $\tau=\{2, 3, 4\}$ and when $\sigma =\{7\}$ and $\tau =\{5, 6\}$.  Since $x_1x_2x_3x_4 \notin J_{\C_7}$, the link $\Lk_1(\Delta|_{1234})$ contains the top-dimensional face $234$, and thus is a simplex, which is contractible (see Figure~\ref{fig:simplex-links}A).  Similarly, $\Lk_7(\Delta|_{567})$ (shown in Figure~\ref{fig:simplex-links}B) is a simplex since $x_5x_6x_7 \notin J_{\C_7}$, and so is contractible.  Signature C-1 precisely characterizes when all the relevant links $\Lk_\sigma(\Delta|_{\sigma \cup \tau})$ are simplices, and hence contractible, for minimal receptive field relationships $(\sigma, \tau)$.  This ensures the absence of any $\CF$-detectable local obstructions.  The previous example showed that some codes satisfying signature C-1 are convex; however, this signature does \underline{not} guarantee convexity.  Specifically, code $\C_2$ from Example~\ref{ex:diff-generalizations} also satisfies this signature, and in fact has no local obstructions, yet that code is \underline{not} convex \cite{counterexample}.

\begin{example}[signature C-2]\label{ex:C-2}
Consider the code
\begin{eqnarray*}
&\C_8=&\{0000000, 0100000, 0010000, 0001000, 0000100, 0000010,  \\
&&1100000, 1010000, 1001000, 0101000, 0011000, 0010100, 0001100,  \\
&& 0000110, 0010001, 0001001, 0000101, 0000011, 1101000, 1011000,  \\
&& 0011100, 0011001, 0010101, 0001101, 0000111, 0011101\}.
\end{eqnarray*}
More compactly, 
$$\C_8= \{\emptyset, 2, 3, 4, 5, 6, 12, 13, 14, 24, 34, 35, 45, 56, 37, 47, 57, 67, 124, 134, 345, 347, 357, 457, 567, 3457\}.$$
This code has
\begin{eqnarray*}
&&\CF^1(J_{\C_8})=\{x_1x_5, \ x_1x_6, \  x_1x_7, \ x_2x_3, \ x_2x_5, \ x_2x_6, \ x_2x_7, \ x_3x_6, \ x_4x_6\} \hspace{.03in}\textrm{ and } \hspace{.03in}  \\
&& \CF^2(J_{\C_8})=\{ x_1(1-x_2)(1-x_3)(1-x_4), \ x_7(1-x_3)(1-x_4)(1-x_5)(1-x_6)\}.
\end{eqnarray*}
For $x_1(1-x_2)(1-x_3)(1-x_4) \in \CF^2(J_{\C_8})$, we have $\sigma=\{1\}$ and $\tau=\{2, 3, 4\}$.  Observe that $x_4$ does not appear together with $x_1$, $x_2$, or $x_3$ in $\CF^1(J_{\C_8})$, so for $i=4 \in \tau$, we have $x_ix_\omega \notin \CF^1(J_{\C_8})$ for every $\omega \subseteq \sigma \cup \tau$.  

For $x_7(1-x_3)(1-x_4)(1-x_5)(1-x_6) \in \CF^2(J_{\C_8})$, we have $\sigma =\{7\}$ and $\tau =\{3, 4, 5, 6\}$.  Observe that $x_5$ does not appear with any of $x_3$, $x_4$, $x_6$, or $x_7$ in $\CF^1(J_{\C_8})$, so for $i=5 \in \tau$, we have $x_ix_\omega \notin \CF^1(J_{\C_8})$ for every $\omega \subseteq \sigma \cup \tau$.   Thus, signature C-2 is satisfied, and so $\C_8$ has no $\CF$-detectable local obstructions.  In fact, $\C_8$ is convex, as it is max $\cap$-complete.

\begin{figure}[!ht]
\begin{centering}
\includegraphics[height=1in]{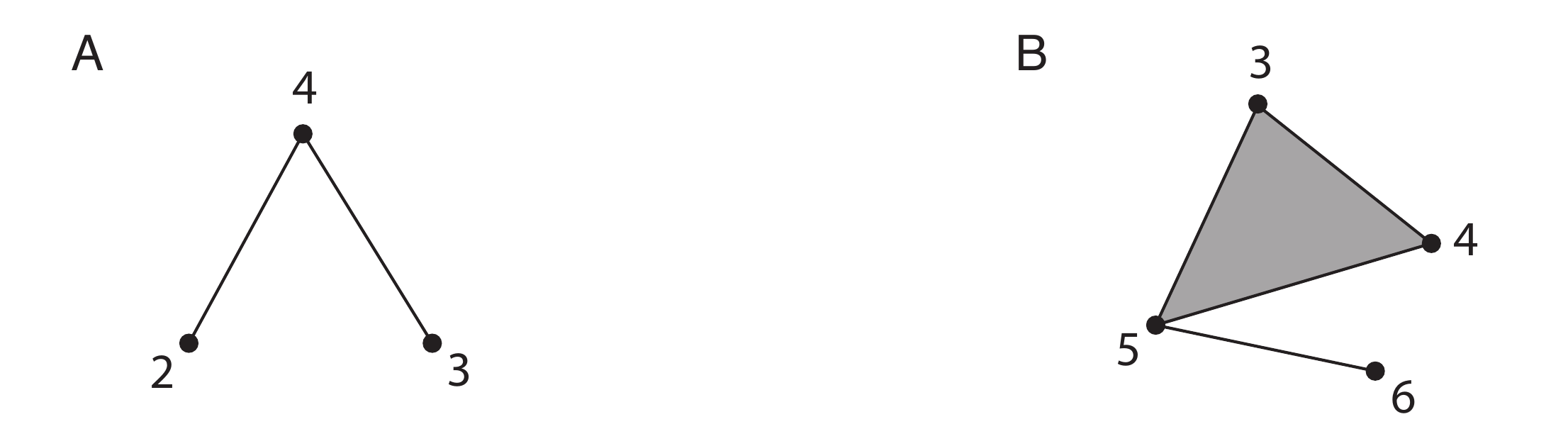}
\caption{Links from Example~\ref{ex:C-2}.  (A) $\Lk_1(\Delta|_{1234})$ is a cone with respect to 4.  (B) $\Lk_7(\Delta|_{34567})$ is a cone with respect to vertex 5.}
\label{fig:cone-links}
\end{centering}
\end{figure} 

\end{example}

As noted in Example~\ref{ex:C-1}, to understand the absence of $\CF$-detectable local obstructions we need to consider links for pairs $(\sigma, \tau)$ where $x_\sigma \prod_{i \in \tau}(1-x_i) \in \CF^2(J_\C)$.  For $\C_8$, the relevant links occur when $\sigma=\{1\}$ and $\tau=\{2, 3, 4\}$ and when $\sigma =\{7\}$ and $\tau =\{3, 4, 5, 6\}$.  The link $\Lk_1(\Delta|_{1234})$ is shown in Figure~\ref{fig:cone-links}A and is a cone with respect to vertex 4, so is contractible.  Similarly, $\Lk_7(\Delta|_{34567})$ is shown in Figure~\ref{fig:cone-links}B and is a cone with respect to vertex 5, so is contractible.  In fact, signature C-2 characterizes when all the relevant links $\Lk_\sigma(\Delta|_{\sigma \cup \tau})$ are cones, and hence contractible, for minimal receptive field relationships $(\sigma, \tau)$.  Thus, signature C-2 generalizes C-1.  This again ensures the absence of any $\CF$-detectable local obstructions, but does not necessarily ensure convexity (e.g.\ code $\C_2$ from Example~\ref{ex:diff-generalizations} satisfies this signature, but is not convex).  

It is worth noting though that $\cap$-complete codes (characterized by signature B-4) are a special class of codes satisfying signature C-2 that are guaranteed to be convex.  Specifically, if a code satisfies B-4, then every element of $\CF^2(J_\C)$ has the form $x_\sigma (1-x_i)$, and so Lemma~\ref{lemma:J_C-algebra} guarantees that $x_\sigma x_i \notin J_\C$.  Thus, no factor of $x_\sigma x_i$ can be in $\CF^1(J_\C)$, and so signature C-2 holds.

\section*{Acknowledgements}
This work began at a 2014 AMS Mathematics Research Community,  ``Algebraic and Geometric Methods in Applied Discrete Mathematics,'' which was supported by NSF DMS-1321794.
CC was supported by NIH R01 EB022862 and NSF DMS-1516881; EG was supported by NSF DMS-1620109; JJ was supported by NSF DMS-1606353; KM was supported by NIH R01 EB022862; and AS was supported by NSF DMS-1312473/1513364 and Simons Foundation grant 521874.  We thank Mohamed Omar and Caitlin Lienkaemper for numerous discussions.

\bibliographystyle{plain}
\bibliography{convexity-refs}

\section{Supplemental Text}

\subsection*{S1: Computing the canonical form $\CF(J_\C)$}  In the following two examples, we illustrate how to compute the canonical form by hand.  For details on how to algorithmically calculate $\CF(J_\C)$ and software to support this, see \cite{NeuralIdealsSage}.

\begin{example}[$\CF(J_{\C_5})$ from Example~\ref{ex:B-2}]
Consider the code $$\C_5=\{0000, 1000, 0100, 0010, 0001, 1100, 1010, 0110, 0011, 1110\}$$ from Example~\ref{ex:B-2}.  Here we show how to compute $\CF(J_{\C_5})$ by hand.

Recall the neural ideal $J_\C \od \langle \chi_{\vv}~|~\vv \in \F_2^n \setminus \C \rangle,$ where $\chi_{\vv}$ is the characteristic pseudo-monomial of $\vv$ (as defined in Equation~\eqref{eq:char-pseudo-monomials}).  The non-codewords are 1001, 0101, 1101, 1011, 0111, 1111, and so 
$$\begin{array}{l} J_{\C_5}=\langle\{x_1x_4(1-x_2)(1-x_3), \ x_2x_4(1-x_1)(1-x_3),\\
 \hspace{.54in} x_1x_2x_4(1-x_3), \ x_1x_3x_4(1-x_2), \ x_2x_3x_4(1-x_1), \ x_1x_2x_3x_4 \}\rangle. \end{array}$$
Since $x_1x_4$ is a minimal divisor of generators in $J_{\C_5}$ that vanishes on all codewords (so it is in $J_{\C_5}$), we have $x_1x_4 \in \CF^1(J_{\C_5})$.  Similarly, $x_2x_4\in \CF^1(J_{\C_5})$.  
Since all the generators of $J_{\C_5}$ are multiples of $x_1x_4$ and $x_2x_4$, both of which are monomials, it follows that $\CF^2(J_{\C_5})$ is empty and $\CF^1(J_{\C_5})=\{x_1x_4, \ x_2x_4\}$.  Thus, $\CF(J_{\C_5})=\CF^1(J_{\C_5})~\cup~\CF^2(J_{\C_5})$ where 
$$\CF^1(J_{\C_5})=\{x_1x_4, \ x_2x_4\} \quad   \textrm{and} \quad  \CF^2(J_{\C_5})=\emptyset.$$
\end{example}

\begin{example}[$\CF(J_{\C_6})$ from Example~\ref{ex:B-3}]
Consider the code $\C_6=\{0000, 0100, 0001, 1100, 1010, 1110\}$ from Example~\ref{ex:B-3}.  The non-codewords are 1000, 0010, 1001, 0110, 0101, 0011, 1101, 1011, 0111, 1111.  Thus, 
$$\begin{array}{l} J_{\C_6}=\langle\{x_1(1-x_2)(1-x_3)(1-x_4), \ x_3(1-x_1)(1-x_2)(1-x_4), \\
\hspace{.54in} x_1x_4(1-x_2)(1-x_3), \ x_2x_3(1-x_1)(1-x_4), \ x_2x_4(1-x_1)(1-x_3),\\
\hspace{.54in} x_3x_4(1-x_1)(1-x_2), \ x_1x_2x_4(1-x_3), \ x_1x_3x_4(1-x_2), \ x_2x_3x_4(1-x_1), \ x_1x_2x_3x_4 \}\rangle. \end{array}$$
Since $x_1x_4$ is a minimal divisor of generators in $J_{\C_6}$ that vanishes on all codewords, $x_1x_4 \in \CF^1(J_{\C_6})$.  Similarly, $x_2x_4, x_3x_4 \in \CF^1(J_{\C_6})$.  Since $1110 \in \C_6$, none of $x_1x_2$, $x_1x_3$, nor $x_1x_2x_3$ is in $J_{\C_6}$, and so $\CF^1(J_{\C_6})=\{x_1x_4, \ x_2x_4, \ x_3x_4\}$.  Every pseudo-monomial in $J_{\C_6}$ is a multiple of one of the monomials in $\CF^1(J_{\C_6})$ except for $x_1(1-x_2)(1-x_3)(1-x_4), \ x_3(1-x_1)(1-x_2)(1-x_4),$ and  $x_2x_3(1-x_1)(1-x_4)$.  The minimal pseudo-monomials in $J_{\C_6}$ that generate these are $x_3(1-x_1)$ and $x_1(1-x_2)(1-x_3)$, and so $\CF^2(J_{\C_6}) \supseteq \{x_3(1-x_1), \ x_1(1-x_2)(1-x_3)\}$.   In fact, one can check that this is the complete set of generators of $\CF^2(J_\C)$ \cite{neural_ring}.  
Thus, $\CF(J_{\C_6})=\CF^1(J_{\C_6})~\cup~\CF^2(J_{\C_6})$ where 
$$\CF^1(J_{\C_6})=\{x_1x_4, \ x_2x_4, \ x_3x_4\} \quad   \textrm{and} \quad  \CF^2(J_{\C_6})=\{x_3(1-x_1), \ x_1(1-x_2)(1-x_3)\}.$$
\end{example}

\subsection*{S2: $\CF$-dectectable local obstructions}
Some local obstructions $(\sigma,\tau)$ correspond to minimal RF relationships.  Among those that do not, we distinguish local obstructions that can be ``stripped down'' (by removing neurons from $\sigma$ and/or $\tau$) to local obstructions corresponding to minimal RF relationships.  We refer to local obstructions that correspond to minimal RF relationships or that can be stripped down to such as \emph{$\CF$-detectable} local obstructions (precise definition was given in Definition \ref{def:CF-detectable}).  As we will see, both these types of local obstructions can be detected directly from $\CF(J_\C)$. 

Since every minimal RF relationship corresponds to a pseudo-monomial in $\CF(J_\C)$ (Lemma~\ref{lemma:minRF}), Lemma~\ref{lemma:A-signature} shows that all $\CF$-detectable local obstructions can be determined solely from the canonical form.  

\begin{lemma}\label{lemma:A-signature}
Given a code $\C$, the following are equivalent:
\begin{enumerate}[(1)]
\item The link $\Lk_\sigma(\Delta|_{\sigma \cup \tau})$ is contractible for every $(\sigma, \tau)$ such that $x_\sigma \prod_{i \in \tau} (1-x_i) \in \CF^2(J_\C)$,
\item The link $\Lk_\sigma(\Delta|_{\sigma \cup \tau})$ is contractible for every minimal RF relationship $(\sigma, \tau)$ with $\tau \neq \emptyset$, and
\item $\C$ has no $\CF$-detectable local obstructions.
\end{enumerate}
\end{lemma}

\begin{proof}
It is clear that $(1)$ and $(2)$ are equivalent since $x_\sigma \prod_{i \in \tau} (1-x_i) \in \CF^2(J_\C)$ if and only if $(\sigma,\tau)$ is a minimal RF relationship with $\tau \neq \emptyset$ by Lemma~\ref{lemma:minRF} and the definition of $\CF^2(J_\C)$.  

\noindent We now prove $(2) \Leftrightarrow (3)$ by contrapositive.  If $\C$ has a $\CF$-detectable local obstruction $(\sigma, \tau)$, then by definition there exist $\sigma' \subseteq \sigma$ and $\tau' \subseteq \tau$ such that $(\sigma', \tau')$ is a minimal RF relationship that gives a local obstruction.  Thus for that $(\sigma', \tau')$, $\Lk_{\sigma'}(\Delta|_{\sigma' \cup \tau'})$ is not contractible, and so (2) does not hold.  Conversely, if there exists a minimal RF relationship $(\sigma, \tau)$ such that $\Lk_\sigma(\Delta|_{\sigma \cup \tau})$ is not contractible, then $(\sigma, \tau)$ is itself a $\CF$-detectable local obstruction, and so (3) does not hold.  
\end{proof}

\begin{theorem}\label{thm:A-sigs-no-CF-detectable}
If $\C$ has either of the algebraic signatures in rows C-1 or C-2 of Table~\ref{table:A-signatures-CF-detectable}, then $\C$ has no $\CF$-detectable local obstructions.  

\begin{table}[!h]
\begin{center}
\begin{small}
\begin{tabular}{l l c l}
& Algebraic signature of $J_\C$ & & Property of $\C$\\
 \specialrule{.125em}{.2em}{.2em} 
\multirow{2}{*}{C-1} & \multirow{2}{*}{$\forall \; x_\sigma\prod_{i \in \tau}(1-x_i) \in \CF^2(J_\C)$, $x_\sigma x_\tau \notin J_\C$} & \multirow{2}{*}{$\Rightarrow$}  & no $\CF$-detectable\\ 
& & & local obstructions\\
 \specialrule{.05em}{.2em}{.2em} 
\multirow{2}{*}{C-2} & $\forall \; x_\sigma\prod_{i \in \tau}(1-x_i) \in \CF^2(J_\C)$,&\multirow{2}{*}{$\Rightarrow$} &  no $\CF$-detectable\\
 &  $\exists \; i\in \tau$ s.t. $x_ix_\omega \notin \CF^1(J_\C)$ for all $\omega \subseteq \sigma\cup \tau$ & & local obstructions \\ 
 \specialrule{.125em}{.2em}{.2em} 
\end{tabular}
\caption{Algebraic signatures of codes with no $\CF$-detectable local obstructions.  These codes are \underline{not} guaranteed to be convex or non-convex, but do not have any local obstructions that can be detected from the canonical form.}
\label{table:A-signatures-CF-detectable}
\end{small}
\end{center}
\end{table}
\vspace{-.3in}
\end{theorem}

\begin{proof}
(C-1) Observe that $x_\sigma x_\tau \notin J_\C$ implies that $x_\sigma x_{\tau'} \notin J_\C$ for all $\tau' \subseteq \tau$ since $J_\C$ is an ideal.  Thus, $\tau' \in \Lk_\sigma(\Delta|_{\sigma \cup \tau})$ for all $\tau' \subseteq \tau$, and so $\Lk_\sigma(\Delta|_{\sigma \cup \tau})$ is the full simplex on the vertex set $\tau$.  Thus $\Lk_\sigma(\Delta|_{\sigma \cup \tau})$ is contractible.  Since this holds for all $(\sigma, \tau)$ such that $x_\sigma \prod_{i \in \tau} (1-x_i) \in \CF^2(J_\C)$, Lemma~\ref{lemma:A-signature} guarantees that $\C$ has no $\CF$-detectable local obstructions.

(C-2)  We will show that signature C-2 guarantees that for every $(\sigma, \tau)$ with $ x_\sigma\prod_{i \in \tau}(1-x_i) \in \CF^2(J_\C)$, the link $\Lk_\sigma(\Delta|_{\sigma \cup \tau})$ is a cone, and hence is contractible.  Consider $\tilde \tau \in \Lk_\sigma(\Delta|_{\sigma \cup \tau})$, so that $\tilde\tau \subseteq \tau$ and $\sigma \cup \tilde\tau \in \Delta(\C)$.  Since $\sigma \cup \tilde\tau \in \Delta(\C)$, we have $x_\sigma x_{\tilde\tau} \notin J_\C$.  By hypothesis, there exists an $i$ such that for every $\sigma' \subseteq \sigma$ and $\tau' \subseteq \tilde \tau$, $x_ix_{\sigma'}x_{\tau'} \notin \CF^1(J_\C)$.  Since $\CF^1(J_\C)$ generates the monomials of $J_\C$, this condition together with $x_\sigma x_{\tilde\tau} \notin J_\C$ guarantees that $x_ix_\sigma x_{\tilde\tau} \notin J_\C$, and so $\{i\} \cup \sigma \cup \tilde\tau \in \Delta(\C)$ implying that $\{i\} \cup \tilde \tau \in \Lk_\sigma(\Delta|_{\sigma \cup \tau})$.  Hence $\Lk_\sigma(\Delta|_{\sigma \cup \tau})$ is a cone with respect to $i$, and so is contractible.  Since this holds for all $(\sigma, \tau)$ such that $x_\sigma \prod_{i \in \tau} (1-x_i) \in \CF^2(J_\C)$, Lemma~\ref{lemma:A-signature} guarantees that $\C$ has no $\CF$-detectable local obstructions.
\end{proof}

As mentioned at the end of Section~\ref{sec:examples}, C-1 is just a special case of C-2.  Specifically, if a code satisfies C-1, then every $i \in \tau$ will satisfy the conditions of C-2, since C-1 guarantees that each link is a simplex, and thus also is a cone with any vertex acting as a cone point.  We nevertheless include the proof of C-1 to clarify the structure of these links.

\end{document}